\documentclass{article}

\PassOptionsToPackage{numbers, sort&compress}{natbib}

\usepackage[preprint]{neurips_2025}

\usepackage[utf8]{inputenc} 
\usepackage[T1]{fontenc}    
\usepackage{hyperref}       
\usepackage{url}            
\usepackage{booktabs}       
\usepackage{amsfonts}       
\usepackage{nicefrac}       
\usepackage{microtype}      
\usepackage{xcolor}         

\usepackage{Definitions}

\usepackage[linesnumbered, ruled, vlined]{algorithm2e}
\usepackage{subcaption}
\usepackage{graphicx}       
\graphicspath{{fig/}}     

\usepackage{mathtools}

\title{Capacitated Fair-Range Clustering: Hardness and Approximation Algorithms\thanks{Authors listed in alphabetical order.}}

\author{%
   Ameet Gadekar \\
   CISPA Helmholtz Center for Information Security\\
   Saarbr{\"u}cken, Germany\\
   \texttt{{firstname.lastname@cispa.de}}
   \And
   Suhas Thejaswi \\
   Max Planck Institute for Software Systems \\
   Kaiserslautern, Germany \\
   \texttt{{lastname@mpi-sws.org}}
}

\begin{document}

\maketitle

\begin{abstract}
Capacitated fair-range $k$-clustering generalizes classical $k$-clustering by incorporating both capacity constraints and demographic fairness. In this setting, data points are categorized as clients and facilities, where each facility has a capacity limit and may belong to one or more possibly intersecting demographic groups. The task is to select $k$ facilities as centers and assign each client to a center such that: ($a$) no center exceeds its capacity, ($b$) the number of centers selected from each group lies within specified lower and upper bounds---defining the fair-range constraints, and ($c$) the clustering cost---\eg, $k$-median or $k$-means---is minimized.

Prior work by Thejaswi et al. (KDD 2022) showed that even satisfying fair-range constraints is \np-hard, thereby making the problem inapproximable to any polynomial factor. We strengthen this result by showing that inapproximability persists even when the fair-range constraints are trivially satisfiable, highlighting the intrinsic computational complexity of the clustering task itself. Assuming standard complexity-theoretic conjectures, we further show that no non-trivial approximation is possible without exhaustively enumerating all $k$-subsets of the facility set. Notably, our inapproximability results hold even on tree metrics and even when the number of groups is logarithmic in the size of the facility set.

In light of these strong inapproximability results, we focus our attention to a more practical setting where the number of groups is constant. In this regime, we design two approximation algorithms: ($i$) a polynomial-time $\bigO(\log k)$- and $\bigO(\log^2 k)$-approximation algorithm for the $k$-median and $k$-means objectives, and ($ii$) a fixed-parameter tractable algorithm parameterized by $k$, achieving $(3+\epsilon)$- and $(9 + \epsilon)$-approximation, respectively. These results match the best-known approximation guarantees for capacitated clustering without fair-range constraints and resolves an open question posed by Zang et al. (NeurIPS 2024).
\end{abstract}

\section{Introduction}
\label{sec:intro}

Clustering is the task of partitioning a set of data points into clusters by choosing $k$ representative points, known as cluster centers (or simply centers, when the context is clear), and assigning each data point to a center to form a clustering solution.
The quality of a clustering solution is typically measured using clustering cost, most commonly defined by the $k$-median (or $k$-means) objective, where the goal is to minimize the sum of (squared) distances between each data point and its assigned center.
In a more general setting, data points are distinguished as clients and facilities---which may or may not overlap---with a constraint that cluster centers must be chosen from the set of facilities.
Further, in capacitated clustering, each facility is also associated with a capacity that limits the number of clients that can be assigned to it. Here, the task is to choose $k$ centers and assign clients to centers in a way that does not exceed their capacities, while minimizing the clustering cost.

In real-world applications, data points can be associated with attributes such as sex, education level, or language skills, forming possibly intersecting groups corresponding to these attributes. In this setting, and with the growing focus on algorithmic fairness, clustering problems that require selecting centers from different demographic groups have been studied under the umbrella of fair clustering~\cite{chhabra2021anoverview}. This line of work, specifically addressing cluster center fairness, has introduced several problem formulations that impose lower bounds, upper bounds, or equality constraints on the number of centers chosen from each group~\cite{hajiaghayi2012local,krishnaswamy2011matroid,kleindessner2019fair,thejaswi2021diversity,thejaswi2022clustering} A more general variant, known as fair-range clustering that has both lower and upper bounds on the number of centers chosen from each group~\cite{hotegni2023approximation,zhang2024parameterized,thejaswi2024diversity}. While prior efforts have primarily focused on fairness in uncapacitated settings, many real-world applications often impose capacity limitations for cluster centers, which is the focus of our work.

To further motivate the relevance of studying this setting, consider a university mentorship initiative to support incoming students from diverse academic, socioeconomic, and cultural backgrounds. The program aims to assign each student (client) to a mentor (facility) who will serve as their primary point of contact for guidance. Each mentor has a limited capacity---they can support only a fixed number of students due to time constraints---and mentors belong to one or more demographic groups---\eg, based on sex, country of origin, or academic discipline---forming possibly intersecting groups. To ensure the program to be effective, the university should solve a clustering task: ($i$) assigning students to mentors based on shared academic goals or proximity in fields of study (minimizing a clustering objective), ($ii$) respecting mentor capacity limits, and ($iii$) ensuring diversity in mentor selection---\eg, ensuring representation from women, international faculty, or underrepresented scientific disciplines.

This example highlights a broader class of real-world (clustering) problems where diversity, capacity, and proximity must all be considered when designing algorithmic decision-support systems. Such problems can be formalized as the \emph{capacitated fair-range clustering} problem, where the goal is to select $k$ centers from a set of facilities and assign each client to a center such that the number of clients assigned to each center does not exceed its capacity (\emph{capacity constraints}), and ensure that the number of centers selected from each group lies within specified lower and upper bounds (\emph{fair-range constraints}). The clustering objective can be $k$-median or $k$-means, resulting in the \emph{capacitated fair-range $k$-median} or \emph{capacitated fair-range $k$-means} problem.

In light of the growing interest in fair clustering, there has been remarkable progress towards understanding the computational complexity as well as design of algorithms for these problems, both in polynomial-time and fixed-parameter tractable (\fpt) setting.%
\footnote{Informally, a (parameterized) problem $P$ is fixed-parameter tractable (approximable) if there exists an algorithm that for any instance $(x,k) \in P$  computes an exact (approximate) solution in time  $f(k) \cdot |x|^{\bigO(1)}$, for some computable $f$;  $k$ is called the parameter of the problem. We denote by $\fpt(k)$ for such running times. }
When the groups are disjoint, polynomial-time approximation algorithms are known for fair-range clustering~\cite{hotegni2023approximation}. However, when the groups intersect, Thejaswi et al.~\cite{thejaswi2021diversity} showed\footnote{In fact, their reduction produces instances with lower-bound only requirements. However, our results can be extended to produce instances with lower-bound only requirements. See Appendix~\ref{app:hardness} for details.} that the problem is inapproximable to any multiplicative factor. A key insight in their result is that, with intersecting groups, even satisfying the fair-range constraints becomes \np-hard regardless of the clustering objective to be optimized. As a consequence, the fair-range $k$-median ($k$-means) problem is inapproximable to any multiplicative factor, both in polynomial-time and in \(\fpt(k)\)-time, even for structured inputs such as Euclidean and tree metrics. Naturally, these results extend to the capacitated variants of these fair-range clustering problems, as they capture the corresponding uncapacitated versions.

While their inapproximability result is significant, it falls short to capture the true complexity of the underlying clustering task, as it focuses solely on the hardness of satisfying the fair-range constraints. In practice, there exist many instances---including those with intersecting groups---where a feasible solution (\ie, one satisfying fair-range constraints) can be found efficiently (or in polynomial-time). For example, a simple greedy strategy that selects facilities covering the most constraints may produce a feasible solution. However, such solutions can be arbitrarily far from being optimal in terms of the clustering cost. To further strengthen the complexity landscape of this problem, we ask: 
\begin{tcolorbox}
{\bf Question:} \emph{Is it possible to approximate the (capacitated) fair-range clustering problem when feasible solutions can be found in polynomial-time?}
\end{tcolorbox}
In this work, we answer this question negatively, revealing the intrinsic hardness of the underlying clustering problem. Additionally, we identify instances that are of practical interest but bypass the above hardness result, and design polynomial-time and $\fpt(k)$-time approximation algorithms. In detail, our contributions are  as follows.\footnote{All proofs are available in the Appendix.} We use $n$ to denote the number of data points in the instance.

\xhdr{Hardness of Approximation}
We strengthen the inapproximability landscape by showing that the hardness does not arise solely from the complexity of satisfying the fair-range constraints. Specifically, we prove that the fair-range $k$-median (and $k$-means) problem remains \np-hard to approximate to any polynomial factor, even when feasible solutions can be found in polynomial-time. While our inapproximability factor matches that of \citet{thejaswi2021diversity}, our result is fundamentally stronger, as the hardness arises from the underlying clustering task itself (see Theorem~\ref{thm:hard:np1} for a precise statement). 
Since capacitated variants generalize their uncapacitated counterparts, our inapproximability results naturally extend to the capacitated setting. We further strengthen our hardness result in two ways. First, observe that any feasible solution, which can be found efficiently in this case, is a $\Delta$ (or $\Delta^2$) approximate solution for fair-range $k$-median (or $k$-means), where $\Delta$ is the distance aspect ratio of the instance.%
\footnote{\label{foot:aspect-ratio} In a metric space $(X,d)$, the aspect ratio $\Delta$ is the ratio between the maximum and minimum pairwise distances, \ie, $\Delta := \frac{d_{\max}}{d_{\min}}$, where $d_{\max}=\max_{x,y \in X}d(x,y)$ and $d_{\min}=\min_{x,y \in X}d(x,y)$.}
In stark contrast, we show that this factor is essentially optimal under $\p \neq \np$ conjecture (see Theorem~\ref{thm:hard:npdelta} for details).
Next, assuming \GAPETH
\footnote{Roughly speaking, \gapeth says that there exists an $\epsilon>0$ such that there is no $2^{o(n')}$ time algorithm that decides if the given \threesat formula $\phi$ on $n'$ variables has a satisfying assignment or every assignment satisfies at most $(1-\epsilon)$ fraction of clauses of $\phi$. See Hypothesis~\ref{hyp: gapeth} for a precise formulation.\label{foot:gapeth}},
we show a stronger result (see Theorem~\ref{thm:hard:enu1}): there is no $n^{o(k)}$-time algorithm
that can approximate the (capacitated) fair-range $k$-median (or $k$-means) problem to any polynomial
factor, even when feasible solutions can be found in polynomial-time. 
Note that the trivial
brute-force algorithm, which enumerates all $k$-tuples of facilities, runs in time $n^{\bigO(k)}$.
Our hardness result implies that this is essentially the best possible---even when seeking only an
approximate solution. Furthermore, our inapproximability result holds even when the number of groups
is logarithmic in the size of the facility set, and even on tree metrics.

\xhdr{Approximation Algorithms}
In light of strong inapproximability results, we turn our attention to identifying instances, for which we can obtain non-trivial approximations. One regime that bypasses the above theoretical hardness barrier, and is simultaneously of practical interest is when the number of groups is constant. This setting has been extensively studied in prior work~\cite{kleindessner2019fair,thejaswi2021diversity,thejaswi2022clustering,zhang2024towards}.
In this setting, we show that design of non-trivial factor approximation algorithms are indeed possible.

\textit{Polynomial-time approximation algorithms.}
For constant many groups,  we present $\bigO(\log k)$- and $\bigO(\log^2 k)$-approximation algorithms for the $k$-median and $k$-means objectives, respectively (see Theorem~\ref{thm:polyapx}). Our algorithms run in polynomial-time and match the best-known approximation factors for their non-fair counterparts~\cite{charikar1998rounding}.
Our approach relies on embedding the original instance into a tree metric, followed by solving the problem exactly on the tree using dynamic programming. Such tree embeddings are well-studied~\cite{bartal1996probabilistic,bartal1998onapproximating,fakcharoenphol2004atight}, and have been applied to obtain approximation algorithms for clustering problems~\cite{charikar1998rounding,bartal1998onapproximating,adamczyk2019constant}, among other optimization problems. However, naively embedding all data points into a tree yields $\bigO(\log n)$-approximation ($\bigO(\log^2 n)$ resp.), since these embeddings suffer from $\bigO(\log n)$ distortion in the distances. Our approximation algorithms achieve significantly better approximation factors, \textit{viz.}, $\bigO(\log k)$ and $\bigO(\log^2 k)$ factors for the $k$-median and $k$-means objectives, respectively.

\textit{Constant-factor $\fpt(k)$-approximation algorithms.}
In pursuit of constant-factor approximation algorithms, we explore the \fpt regime with respect to parameter $k$, the number of centers in the solution. While our inapproximability result rules out $n^{o(k)}$-time approximation algorithms in the general setting, this hardness result no longer applies when the number of groups is constant.
As our next contribution, in Theorem~\ref{thm:fptapx}, we give $(3 + \epsilon)$ and $(9 + \epsilon)$-approximation algorithms, for any $\epsilon >0$, for the capacitated fair-range $k$-median and $k$-means problems, respectively. These algorithms run in time $(\bigO(k \epsilon^{-1} \log n))^{\bigO(k)} \cdot n^{\bigO(1)}$,  for constant number of groups, and match the best-known approximation guarantees for their unfair counterparts~\cite{cohenaddad2019on}.
Our algorithm is based on the leader-guessing framework, which has been successfully applied to solve several clustering problems in recent years~\cite{cohenaddad2019on, cohenaddad2019tight,thejaswi2022clustering,zhang2024parameterized,chen2024parameterized}.
A key challenge, however, in directly applying this framework is that the chosen facilities may be infeasible, since they must simultaneously satisfy both capacity  and fairness constraints---which  prior approaches are not equipped to handle.

The rest of the paper is organized as follows: Section~\ref{sec:related} reviews related work, Section~\ref{sec:problem} defines the problem, Section~\ref{sec:hardness} presents inapproximability results, Section~\ref{sec:alg} describes our approximation algorithms, and Section~\ref{sec:discussion} offers conclusions, limitations and broader impact of our work.

\section{Related work}
\label{sec:related}

Our work builds on prior research in clustering and algorithmic fairness. For comprehensive surveys on clustering and fair clustering, we refer the reader to these surveys~\cite{jain1999data, chhabra2021anoverview}.

Clustering is a fundamental problem in computer science, extensively studied in both theoretical and applied domains~\cite{jain1988algorithms,vazirani2001approximation}. Among the most well-known clustering formulations are the $k$-median and $k$-means problems~\cite{vazirani2001approximation}, along with their capacitated variants, where each facility can serve only a limited number of clients~\cite{charikar1998rounding}.
A seminal line of work by \citet{bartal1996probabilistic} introduced approximation algorithms based on probabilistic tree embeddings, yielding an $\bigO(\log^2 n)$-approximation for capacitated $k$-median, and later improved to $\bigO(\log n)$~\cite{fakcharoenphol2004atight}. Despite their practical relevance, the best-known polynomial-time approximations remain at $\bigO(\log k)$ for $k$-median and $\bigO(\log^2 k)$ for $k$-means~\cite{adamczyk2019constant}, with no improvements in recent years.
In the \fpt regime, \citet{adamczyk2019constant} gave a $(7 + \epsilon)$-approximation for capacitated $k$-median in $2^{\bigO(k \log k)} \cdot n^{\bigO(1)}$, and it was later improved to $(3 + \epsilon)$ and $(9 + \epsilon)$ for capacitated $k$-median and $k$-means in  $(\bigO(k \epsilon^{-1} \log n))^{\bigO(k)} \cdot n^{\bigO(1)}$ time~\cite{cohenaddad2019on}.

Fairness in unsupervised machine learning tasks---such as clustering, feature selection, and dimensionality reduction---has gained prominence in recent years as part of a broader focus on algorithmic fairness~\cite{matakos2024fair, gadekar2025fair,kleindessner2019fair,chierichetti2017fair,samadi2018theprice,abbasi2023parameterized}. However, fair clustering was studied even before algorithmic fairness became a prominent research focus. For example, the red-blue median problem limited the maximum number of servers chosen from each type (\eg, red or blue)~\cite{hajiaghayi2012local}, and its generalization, the matroid median problem, captured broader fairness-like constraints~\cite{krishnaswamy2011matroid}. Related problems also appear in robustness-based clustering, which aims to prevent disproportionately high costs for any clients~\cite{bhattacharya2014new}.
Our work focuses on cluster center fairness, which has seen substantial progress in recent years through formulations imposing lower bounds, upper bounds, or equality constraints on the number of centers selected from each group~\cite{gadekar2025fair,kleindessner2019fair,thejaswi2021diversity,thejaswi2022clustering, jones2020fair}. We study the most general formulation---fair-range clustering---which enforces both lower and upper bounds on the number of centers selected from each group.

\citet{hotegni2023approximation} gave a polynomial-time approximation algorithm for the uncapacitated fair-range clustering with disjoint groups under $(\ell,p)$-norm objective.~\citet{thejaswi2024diversity,thejaswi2022clustering} addressed the case of intersecting groups, giving $(1 + \frac{2}{e} + \epsilon)$- and $(1 + \frac{9}{e} + \epsilon)$-approximations for $k$-median and $k$-means, respectively, in $\fpt(k)$-time, when the number of groups is constant. More recently, \citet{zhang2024parameterized} presented a $(1+\epsilon)$-approximation for fair-range $k$-median in Euclidean metrics in $\fpt(k)$-time, and asked about the possibility of designing $\fpt$-approximation algorithms when facilities have capacity constraints. 
\citet{quy2021fair} studied fair clustering with capacity constraints, but their setting differs to us in two ways: first, fairness is imposed on clients via proportional fairness, and capacities limit the size of each cluster. In contrast, we impose fairness on center selection with lower and upper bounds on the number of centers per group. Our capacity limits are tied to facilities---each facility with its own limit---so the cluster size depends on the selected center.

\section{The Capacitated Fair-Range Clustering Problem}
\label{sec:problem}

We formally define of our problem. 

\begin{definition}[The capacitated fair-range $k$-median (and $k$-means) problem]
An instance $\Ical = \CFRkClustIns$ of the \emph{capacitated fair-range $k$-clustering problem} is defined by positive integers $k$ and $t$, a set $C$ of clients, a set $F$ of facilities, and a metric $d$ over $C \cup F$. 
Each facility in $F$  belongs to one or more demographic groups, forming possibly intersecting groups denoted by $\GG=\{G_i\}_{i \in [t]}$. 
Each group $G_i$ is associated with a lower bound requirement $\alpha_i$ and an upper bound requirement $\beta_i$. The requirements are represented by vectors $\alphavec=(\alpha_i)_{i \in [t]}$ and $\betavec=(\beta_i)_{i \in [t]}$. 
Furthermore, each facility $f \in F$ has a capacity $\zeta: F \rightarrow \ZZ_{\geq 0}$.
The task is to select a subset $S \subseteq F$ of at most $k$ facilities and find an assignment function $\rho:C \rightarrow S$ that assigns each client $c \in C$ to a facility $f \in S$, to form a clustering solution $(S,\rho)$. A solution $(S,\rho)$ is feasible if it satisfies:
\squishlist
\item $\forall G_i \in \GG$, the number of selected centers from $G_i$ lie within  $\alpha_i$ and $\beta_i$, \ie, $\alpha_i \leq |S \cap G_i| \leq \beta_i$,
\item $\forall f \in S$, $f$ is assigned at most $\zeta(f)$ clients, \ie, $|\{c\in C: \rho(c) = f\}| \leq \zeta(f)$. 
\squishend
The objective of the \emph{capacitated fair-range $k$-median} is to minimize 
${\cost_\cI(C,S) := \sum_{c \in C} d(c, \rho(c))}$,
while for  capacitated fair-rage $k$-means, the objective is to minimize 
$\cost_\cI(C,S) := \sum_{c \in C} d(c, \rho(c))^2$,
over all feasible solutions $(S,\rho)$.
We succinctly denote these problems as \CFRkMed and \CFRkMeans, respectively.
\end{definition}
When facilities have unlimited capacities and can serve any number of clients, 
the problem is referred as the fair-range $k$-median (and $k$-means) problem and denoted succinctly as \FRkMed (and \FRkMeans). 
When the client set $C$ is clear from context, we write $\cost_{\Ical}(S)$ for $\cost_{\Ical}(C, S)$; when both $\Ical$ and $C$ are clear, we use $\cost(S)$.
For discrete metrics, we assume that $d$ is defined by a weighted graph $H$ whose vertex set  contains $C \cup F$, and where $d$ corresponds to the shortest-path metric on $H$. We say that $d$ is a tree  metric if $H$ is tree.
We use $n$ to denote the size of the vertex set of $H$. For a positive integer $\kappa$, we use $[\kappa] = \{1, \dots, \kappa\}$. We assume the distance aspect ratio, denoted as $\Delta$, of the metric space of the given instance is polynomially bounded, \ie, $\Delta = n^{\bigO(1)}$.

\section{On the Hardness of Approximation of (Capacitated) Fair-Range Clustering}
\label{sec:hardness}

As mentioned in the introduction, the prior work \citet{thejaswi2024diversity} established inapproximability results for \FRkMed (\FRkMeans) without capacity constraints, but their hardness stems from the \np-hardness of satisfying fair-range constraints when groups intersect. However, in many practical settings---including those with intersecting groups---feasible solutions can often be found efficiently, for example by greedily selecting facilities that satisfy the most constraints.
In this section, we show that even when feasible solutions can be found trivially, no polynomial-time algorithm can approximate (capacitated) fair-range clustering to any polynomial factor, assuming $\p \neq \np$. To formalize this, we define \FRkMedO (and \FRkMeansO) as those instances of \FRkMed (and \FRkMeans, resp.) which admit a polynomial-time algorithm for finding feasible solutions. We establish the following hardness of approximation result for these instances.

\begin{restatable}{theorem}{nphardpoly}
\label{thm:hard:np1}
    There is no polynomial-time algorithm that can approximate \FRkMedO (or \FRkMeansO) to any polynomial factor, unless $\p = \np$. The hardness holds even on tree metrics.
\end{restatable}

We further strengthen this result in two ways: first,  we show a stronger inapproximability factor for polynomial time algorithms, and second, we rule out super polynomial time algorithms  with the same inapproximability guarantee as in Theorem~\ref{thm:hard:np1}.
For the former, note that, any feasible solution is a $\Delta$-approximation, where recall that $\Delta$ is the distance aspect ratio of $d$. In the following theorem, we show that this trivial bound is optimal.

\begin{theorem}[Informal version of Theorem~\ref{thm:hard:np2}]\label{thm:hard:npdelta}
    Assuming $\p \neq \np$, there is no polynomial time algorithm to approximate \FRkMedO  to a factor within $(\Delta-2)/16$. Furthermore, the hardness result holds on tree metrics.
\end{theorem}

Next, we strengthen the running time of Theorem~\ref{thm:hard:np1}.
Note that (capacitated) \FRkMed and \FRkMeans can be  exactly solved in time $n^{O(k)}$, by enumerating all $k$-sized subsets of the facility set. Our next result shows that, under $\gapeth^{\ref{foot:gapeth}}$,
even for finding a non-trivial approximation requires \( n^{\Omega(k)} \) time, making brute-force algorithm essentially our best hope, even in this case. This also rules out any $\fpt(k)$-time approximation algorithms for these problems.

\begin{restatable}{theorem}{enuhard}\label{thm:hard:enu1}
    Assuming \gapeth, there is no $f(k)\cdot n^{o(k)}$ algorithm, for any computable function $f$, that can approximate \FRkMedO (or \FRkMeansO) to any polynomial factor. The hardness result holds even when the number of groups is $\bigO(k^3 \log n)$, and even when the metric space is a tree.
\end{restatable}

\begin{remark}
Our inapproximability results are stated for the range setting, where both upper and lower bounds are specified for each group. In contrast, Thejaswi et al.~\cite{ thejaswi2021diversity} show hardness results even when only lower bounds are present, making their result appear formally stronger. However, we note that our hardness constructions can be adapted to obtain the same inapproximability guarantees under lower-bound-only constraints as well. See Appendix~\ref{app:hardness} for details.
\end{remark}

\xhdr{Technical Overview}
Our inapproximability results are based on reductions from the \threesat problem to \FRkMedO (and \FRkMeansO). The high level idea is that (see Theorems~\ref{thm:redthreesat} and~\ref{thm:redgapthreesat}), given an instance $\phi$ of \threesat on $n'$ variables and $m'$ clauses and $D \geq 1$, we construct an instance $\Ical$ of \FRkMed such that $(i)$ the distance aspect ratio  in $\Ical$ is $D$, $(ii$) if $\phi$ is satisfiable, then there is a feasible solution to $\Ical$ with cost at most $k$, and $(iii)$ if every assignment satisfies at most $(1-\epsilon)m$ clauses of $\phi$, for some constant $\epsilon>0$, then every feasible solution to $\Ical$ has cost $\Omega(D)$. 
This immediately rules out $o(D)$-approximation for the problem, for arbitrary values of $D$. The main crux of our construction is that we  use the group structure with requirements to create instances where $D$ can have arbitrary values, without breaking the metric property. This is in contrast with the hardness constructions for the vanilla $k$-Median ($k$-Means) problem, where we do not have such a flexibility.

\begin{wrapfigure}{r}{0.37\textwidth}
	\begin{center}
		\vspace{-20pt}
		\includegraphics[width=0.38\textwidth]{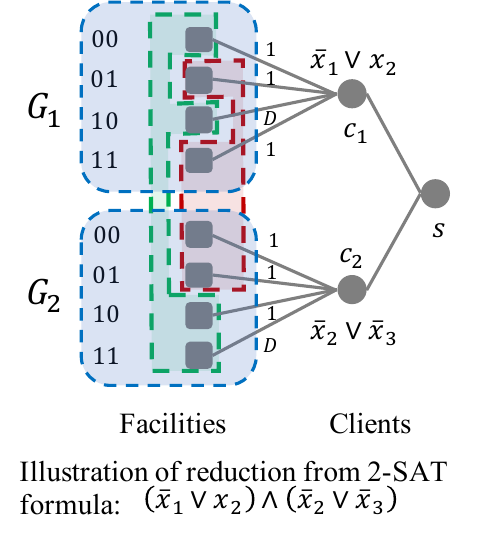}
		\vspace{-15pt}
	\end{center}
	\vspace{-15pt}
\end{wrapfigure}

In more detail, for each clause $C_i$ in $\phi$, we create a client $c_i$ and a set $F_i$ of facilities, each facility representing a possible assignment to the variables in $C_i$. We connect $c_i$ to all facilities in $F_i$, assigning edge weight $1$ if the assignment satisfies $C_i$, and $D$ otherwise. Finally, we add a dummy client $s$ and connect it to each $c_i$ with an edge of weight $D$.
Observe that this creates a tree metric and hence the metric property works for all values of $D$. 
However, note that since we are enumerating all the partial assignments of every clause, we would also like to enforce the constraint that the partial assignments corresponding to the selected facilities in the solution must be consistent. This can precisely be achieved by creating suitable groups and adding corresponding requirements.
In particular, we create two types of groups---\emph{clause groups} and \emph{assignment groups}. For each clause $C_i$ in $\phi$, we create a clause group $G_i$ that contains all the facilities in $F_i$ and set the requirements $\alpha_i = \beta_i=1$, to enforce the selection of exactly one partial assignment for $C_i$. 
To ensure consistency across assignments, we create assignment groups: for each variable $X_j$, each pair of clauses $C_i, C_{i’}$ containing $X_j$, and each assignment $a \in \{0,1\}$, we create a group $G^{(C_i,C_{i’})}_{X_j \mapsto a}$. This group includes all facilities in $F_i$ assigning $a$ to $X_j$, and all in $F_{i’}$ assigning $1-a$ to $X_j$, with both lower and upper bounds set to $1$. Finally, we set $k = m$.
The idea is that any feasible solution to this \CFRkMed instance should correspond to a set of consistent partial assignments that allows us to obtain a global assignment. Furthermore, note that, for every assignment to the variables of $\phi$, the corresponding set of facilities form a feasible solution. Hence, we can find feasible solutions to this instance trivially. 
An illustration of the reduction is shown on the right. Due to space constraints, we depict a reduction from a $2$-SAT formula with two clauses (the construction for \threesat is similar): $C_1=\bar{x}_1 \vee x_2$ and $C_2=\bar{x}_2 \vee \bar{x}_3$. We highlight clause groups $G_1$ and $G_2$ in blue, and two assignment groups---$G^{(C_1,C_2)}_{x_2 \mapsto 0}$ and $G^{(C_1,C_2)}_{x_2 \mapsto 1}$---in green and red, respectively. 

\section{Approximation Algorithms for Constant Number of Groups} 
\label{sec:alg}
In this section, we focus our attention towards a setting that is more practical, but simultaneously avoids the hardness results of the previous section. Specifically, we consider the problem when the number of groups is constant. 
Moreover, this setting  has been extensively explored in the literature (\eg, \cite{kleindessner2019fair,thejaswi2021diversity,thejaswi2022clustering,zhang2024towards}) across various notions of fair clustering. We believe that studying the capacitated fair-range setting under this regime is both natural and promising.
To this end, we present polynomial-time approximation algorithms in Section~\ref{sec:polytime} and $\fpt(k)$-approximation algorithms in Section~\ref{sec:fpttime}.
\subsection{Polynomial-time approximation algorithms}
\label{sec:polytime}

In this subsection, we design polynomial-time $\bigO(\log k)$- and $\bigO(\log^2 k)$-approximation algorithms for \CFRkMed and \CFRkMeans, respectively, when the number of groups is constant. For simplicity, we focus on $\bigO(\log k)$-approximation algorithm for \CFRkMed. Our approach can be easily generalized to \CFRkMeans to obtain $\bigO(\log ^2k)$-approximation (see Appendix~\ref{app:polytime} for details).

\begin{restatable}{theorem}{polyapx} \label{thm:polyapx}
There exists a $\bigO(\log k)$ $(\text{and } \bigO(\log^2 k))$ approximation algorithm for \CFRkMed (\CFRkMeans, resp.) that runs in $(nk^t)^{\bigO(1)}$ time.
\end{restatable}

At a high level, the algorithm proceeds in two steps. First, given an instance $\cI$ of \CFRkMed on general metrics, we embed it into a tree metric. Second, we design a polynomial-time exact dynamic program to solve \CFRkMed on the resulting tree metric.
As mentioned earlier, standard techniques~\cite{bartal1998onapproximating, fakcharoenphol2004atight} allow embedding any metric $\Mcal$ on $n$ points into a tree metric with $\bigO(\log n)$ distortion in the distances.\footnote{The distortion is on expectation over the probabilistic embedding of $\Mcal$ based on a distribution on tree metrics. However, such embeddings can be derandomized. See Appendix~\ref{app:polytime} for details.} Thus, if we can solve \CFRkMed exactly on tree metrics, combining this with the tree embedding yields a $\bigO(\log n)$-approximation for \CFRkMed on general metrics.
To obtain $\bigO(\log k)$-approximation, we build on the ideas of~\cite{adamczyk2019constant}, who designed a $\bigO(\log k)$-approximation algorithm for capacitated $k$-median, extending the techniques from \cite{charikar1998rounding}.

An overview of our approach is shown in Figure~\ref{fig:polytime}. 
In Phase~$1$, we embed the given instance $\Ical$ of \CFRkMed on metric $d$ into a new instance $\Ical'$ on metric $d'$ such that $d'$ dominates $d$,\footnote{That is, $d(u,v) \le d'(u,v)$, for all pairs $u,v \in C \cup F$.} and has properties that enable us to obtain better approximation guarantees. We remark that instances $\Ical'$ and $\Ical$ differ only in the underlying metric.
Specifically, $d'$ corresponds to the shortest-path metric on a graph, consisting of a complete graph (or clique) on $k$ nodes, and remaining $n-k$ nodes connected to exactly one node in the clique. Here $n:= C \cup F$. We refer to this metric as \kmet. To construct such an embedding we make use of a polynomial-time $\bigO(1)$-approximation algorithm $\Acal$ for $k$-Median (\cite{cohenaddad2025a2epsilon, aryalocal, ahmadian2019better}). Below we state our result formally.

\begin{restatable}{lemma}{cliquestaremb} \label{lemma:cliquestaremb}
Given an instance $\Ical$ of \CFRkMed on a general metric $d$, and a polynomial-time $\eta$-approximation algorithm $\Acal$ for $k$-median, we can construct, in $n^{\bigO(1)}$ time, an instance $\Ical'$ of \CFRkMed on \kmet metric $d'$ such that
\[ \cost_{\cI'}(O') \le \cost_{\Ical'}(O) \leq (4 \eta + 3) \cdot \cost_{\Ical}(O), \]
where $O, O' \subseteq F$ are optimal solutions to $\cI$ and $\cI'$, respectively.
\end{restatable}

\begin{figure}
\centering
\footnotesize
\begin{tabular}{c@{\hspace{1.6cm}}c@{\hspace{0.7cm}}c}
     & \textbf{Lemma~\ref{lemma:cliquestaremb}} & \textbf{Lemma~\ref{lemma:treemetric}} \\ [0.1em]
     & $\opt(\Ical') = \bigO(\eta) \cdot \opt(\Ical)$ & $\opt(\Ical'') = \bigO(\log k) \cdot \opt(\Ical')$ \\ [0.5em]
    {\includegraphics[width=0.2\textwidth]{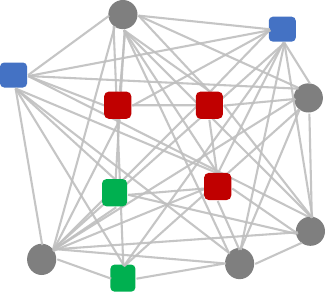}} & 
    {\includegraphics[width=0.2\textwidth]{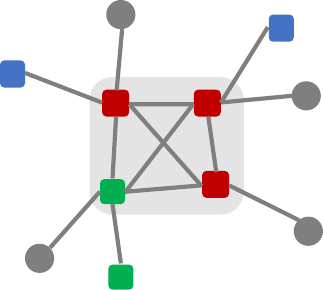}} &
    {\includegraphics[width=0.2\textwidth]{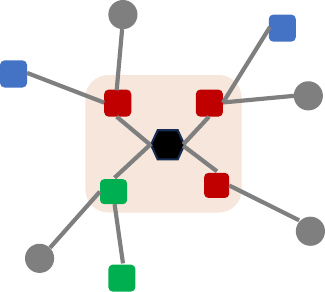}} \\[0.5em]
    \raisebox{0.5\height}{{\bf (a)}} \shortstack[l]{~Instance $\Ical$ \\~Metric $d$} & \raisebox{0.5\height}{{\bf (b)}}\shortstack[l]{~Instance $\Ical'$\\~Clique-star $d'$} & \raisebox{0.5\height}{{\bf (c)}}\shortstack[l]{~Instance $\Ical''$\\~Tree $d''$}
    \end{tabular}
    \caption{\label{fig:polytime} Overview of our algorithm for Theorem~\ref{thm:polyapx}: squares represent facilities, (gray) circles represent clients, and (black) hexagons are dummy nodes introduced in the tree embedding. Colors (red, blue and green) indicate facility groups. 
    Panel (a) shows the original instance $\Ical$ of \CFRkMed in a general metric space $d$. Panel (b) depicts the transformed instance $\Ical'$ in \kmet $d'$ obtained from Lemma~\ref{lemma:cliquestaremb} using an $\eta$-approximation solution $S$ treating $\cI$ as a vanilla $k$-median instance; $S$ is  highlighted with shaded area. Panel (c) illustrates the instance $\Ical''$ in the tree metric $d''$ obtained from Lemma~\ref{lemma:treemetric} 
    with the tree embedding of $S$ again highlighted with shaded area.}
    \vspace{-6mm}
\end{figure}

The idea to build $d'$ using $\Acal$ is as follows. First, we obtain an $\eta$-approximate set $S$ to $\cI$ using $\Acal$. Note, however, that since $\Acal$ works only for $k$-median, the set $S$ may not be feasible (both capacity and fairness wise). In $d''$, we create a clique on the nodes of $S$ with weights of the edges being the distance between the pairs in $d$. Finally, we connect the remaining points to the closest node in $S$ with weight being the corresponding distance in $d$.

In Phase~$2$, we design $\bigO(\log k)$-approximation algorithm for \CFRkMed on \kmet metrics. Towards this, we first replace the clique of $d'$ on $k$ vertices by a tree obtained from tree embeddings mentioned earlier~\cite{fakcharoenphol2004atight}, to obtain a tree metric $d''$. Note that, for any pair $u,v \in C \cup F$, we have $d'(u,v)\le d''(u,v) \le \bigO(\log k)\cdot d'(u,v)$, due to the guarantees of the embedding. In fact, we prove the following stronger result.

\begin{restatable}{lemma}{treemetric}\label{lemma:treemetric}
	Given an instance $\Ical'$ of \CFRkMed of size $n$ on a \kmet metric $d'$, we can construct, in time $n^{\bigO(1)}$, an instance $\Ical''$ of \CFRkMed on a tree metric $d''$ such that for any set $S'$ of facilities, it holds that
	\[ \cost_{\Ical'} (S') \leq  \cost_{\Ical''}(S') \leq  \bigO(\log k) \cdot \cost_{\Ical'}(S'). \]
\end{restatable}

Finally, we show a simple dynamic programming algorithm for \CFRkMed on tree metrics.
\begin{restatable}{lemma}{exacttree}\label{lemma:exact-tree}
There exists an exact algorithm for \CFRkMed on tree metrics in $k^{2t} \cdot n^{\bigO(1)}$ time.
\end{restatable}
We begin by transforming the given tree into a rooted full binary tree where all clients and leaves appear at leaves using standard techniques. We then design a dynamic program over this binary tree.
At a high level, the dynamic programming table $T(e, \kappavec, b)$ stores the minimum cost (and the corresponding solution) over all feasible solutions for the subtree $T_e$ rooted at edge $e$, with respect to $\kappavec$ and $b$. Here, $\kappavec=(\kappa_i)_{i\in [t]}$ specifies that the solution must open $\alpha_i \le \kappa_i \le \beta_i$ facilities  from group $G_i$ in $T_e$, and
$b \in \{-n, \dots, +n\}$  indicates  $|b|$ clients must be routed  through the edge $e$  (routed out of $T_e$ when $b >0$ and routed into $T_e$ when $b <0$). 
To compute this entry,  we  split $\kappa$ and $b$ between left and right subtrees---connected via edges $e^\ell$ and $e^r$---such that $\kappavec = \kappavec^\ell + \kappavec^r$ and $b = b^\ell + b^r$. For each configuration $(e,\kappavec,b)$, we select the tuple $(\kappavec^\ell, \kappavec^r, b^\ell, b^r)$ that minimizes the cost and proceed in bottom-up fashion from the leaves to the root to find an optimal solution.

Combining Lemmas~\ref{lemma:cliquestaremb}, \ref{lemma:treemetric} and \ref{lemma:exact-tree} yields $\bigO(\log k)$- and $\bigO(\log^2 k)$-approximation algorithms for \CFRkMed and \CFRkMeans, respectively, running in $(nk^t)^{\bigO(1)}$ time. 
\subsection{Fixed parameter tractable time approximation algorithms}
\label{sec:fpttime}

In this subsection, we present  $(3+\epsilon)$ and $(9+\epsilon)$-approximation algorithms for \CFRkMed and \CFRkMeans, respectively, that run in $\fpt(k)$-time, for constant number of groups. Our approach is based on the leader-guessing framework, which has been successfully used to obtain \fpt-approximation algorithms for several clustering problems~\cite{cohenaddad2019on,cohenaddad2019tight,thejaswi2022clustering,zhang2024parameterized,chen2024parameterized}.  Our result is formally stated below.
\begin{restatable}{theorem}{fptapx}\label{thm:fptapx}
For any $\epsilon > 0$, there exists a randomized $(3 + \epsilon)$-approximation algorithm for \CFRkMed with running time $(\bigO(2^t k \epsilon^{-1} \log n))^{\bigO(k)} \cdot n^{\bigO(1)}$. With the same running time, a  $(9 + \epsilon)$-approximation algorithm exists for \CFRkMeans.
\end{restatable}

For the sake of clarity, we focus on presenting our algorithm for \CFRkMed.
As a first step towards our goal, we reduce \CFRkMed to several instances of the one-per-group weighted capacitated fair-range  $k$-median problem, abbreviated as \OPGWCkMedDis. In this variant, the facility groups are disjoint, clients are weighted, and the objective is to select exactly one facility from each group to minimize the clustering cost. Our reduction guarantees that  there exists one instance of \OPGWCkMedDis that has the same optimal cost as the \CFRkMed instance. Hence, finding an approximate solution to each of these \OPGWCkMedDis instances, would yield an approximate solution to \CFRkMed.
To perform this transformation, we define a characteristic vector $\chi_f \in \{0,1\}^t$ for each facility $f \in F$, where the $i$-th bit is set to $1$ if $f \in G_i$, and $0$ otherwise. Facilities sharing the same characteristic vector $\gamma$ are grouped into $E(\gamma) = \{f \in F : \chi_f = \gamma\}$. This induces a partition $\EE = \{E(\gamma) : \gamma \in \{0,1\}^t\}$ of the facility set $F$. Since there are at most  $2^t$ distinct characteristic vectors, this results in at most  $2^t$ disjoint facility groups.
In Lemma~\ref{lemma:feasiblecp}, we show that enumerating all possible $k$-multisets $\{\gamma_1, \dots, \gamma_k\}$ of characteristic vectors---referred to as constraint patterns---and checking for the feasibility, \ie, whether they satisfy the fair-range constraints $\alphavec \leq \sum_{i \in [k]} \gamma_i \leq \betavec$ where the inequalities are taken element-wise, can be done in time $2^{tk} \cdot n^{\bigO(1)}$.
\begin{restatable}{lemma}{feasiblecp} \label{lemma:feasiblecp}
 For any instance $\CFRkClustIns$ of \CFRkMed (or \CFRkMeans resp.), there exists a deterministic algorithm to enumerate all feasible constraint patterns in time $2^{tk} \cdot n^{\bigO(1)}$.
\end{restatable}
While choosing exactly one facility from each group $E(\gamma_i)$ gives a feasible solution, it may result in an unbounded approximation ratio. To address this, we design an algorithm that returns a subset of facilities with bounded approximation factor,  for each such instance.
Towards this, we first reduce the number of clients to a small subset of weighted clients $(W, \omega)$, using  coreset, where $W \subseteq C$ and $\omega:W \rightarrow \RR_{\geq 0}$. The guarantee of $(W,\omega)$ is that the cost of every $k$-sized subset of facilities is approximately preserved in the coreset. For capacitated $k$-Median and $k$-Means, such coresets were constructed by~\citet{cohenaddad2019on}. We observe that these constructions also yield coresets for the corresponding fair-range variants, as the fair-range constraints impose restrictions only on the selection of facilities.

By using Lemma~\ref{lemma:feasiblecp} on the coreset (instead of $C$), we construct instances of \OPGWCkMedDis corresponding to all feasible constraint patterns. Each such instance has:
($i$) the facility groups $\EE = \{E(\gamma_i)\}_{i \in [k]}$, derived from each feasible \pattern from Lemma~\ref{lemma:feasiblecp}, and
($ii$) the weighted client set $(W, \omega)$, obtained from the coreset, and
($iii$) the weighted clients in $W$ can be fractionally assigned to the selected facilities $S$, via an assignment function $\mu:W \times S \rightarrow R_{\geq 0}$ such that $\forall c \in C, \sum_{f \in S} \mu(c,f) = \omega(c)$, and the capacity constraints are respected \ie, $\forall f \in S, \sum_{c \in W} \mu(c,f) \leq \zeta(f)$.
While satisfying the fair-range constraints, it may be necessary to select multiple facilities from the same group. To treat such selections as disjoint, we duplicate facilities and place their copies at a small distance $\epsilon_2 > 0$ apart. As a result, the facility sets in the transformed instance become disjoint.
However, this duplication must be handled carefully during facility selection, as choosing multiple copies of the same facility may violate capacities. To address this, we maintain a mapping $M$ that tracks such duplicates and ensures that no original facility is selected more than once. This issue is unique to the capacitated fair-range setting and does not arise in the uncapacitated variant. 
The resulting instance of \OPGWCkMedDis is denoted by $(k, W, \EE, \zeta, \omega)$, where the objective is to choose exactly one facility from each group $E(\gamma_i)$ to minimize the clustering cost, $\fraccost(W,S) := \sum_{c \in W, f \in S} \mu(c,f) d(c,f)$. See Section~\ref{app:fpt-apx} for a formal definition.

Our next step is to design approximation algorithm for \OPGWCkMedDis, which is stated in Lemma~\ref{lemma:opg-fptapx}. 
\begin{lemma} \label{lemma:opg-fptapx}
For any $\epsilon > 0$, there exists a randomized $(3 + \epsilon)$-approximation algorithm for \OPGWCkMedDis in time $(\bigO(k \epsilon^{-1} \log n))^{\bigO(k)} \cdot n^{\bigO(1)}$. Similarly, there exists a randomized  $(9 + \epsilon)$-approximation algorithm for \OPGWCkMeansDis with the same running time.
\end{lemma}
\vspace{-2mm}
To this end, we build upon the leader-guessing framework introduced by \citet{cohenaddad2019tight}, where the algorithm guesses the set of leaders $L^* = \{\ell_i^*\}_{i \in [k]}$ such that each $\ell_i \in W$  is a client closest to the center $f_i^*$ in the optimal solution $S^* = \{f_i^*\}_{i \in [k]}$; $\ell_i$ called the leader of cluster $i$ in $S^*$. 
Additionally, we guess the radii $\{r_i^*\}_{i \in [k]}$ of the leaders, where $r_i$ corresponds to the (approximate) distance of $\ell_i$ to  $f_i^* \in S^*$. Given the leader $\ell_i$ and the corresponding radius $r^*_i$, we carefully choose exactly one facility from $E(\gamma_i)$ within the radius $r_i^*$ of $\ell_i^*$ since the chosen facilities need to satisfy both capacity and fair-range constraints simultaneously.  For instance, it may happen that two facilities $f_i^*, f_j^*$, belonging to different groups, are both located  at a distance $r_i^*$ from $\ell_i$. While $\ell_i$ is the leader for the cluster centered at $f_i^*$, a naive selection might choose $f_j^*$, whose capacity maybe insufficient to serve the intended cluster. 
Moreover, we must avoid choosing a duplicate facility, if the original facility was already chosen. 
Despite these challenges, we show how to handle such cases carefully and   obtain a $(3 + \epsilon)$- and $(9 + \epsilon)$-approximation for \OPGWCkMedDis and \OPGWCkMeansDis, respectively.

For the running time, note that the construction of~\cite{cohenaddad2019on} has
the coreset sizes of $|W| = \bigO(k^2 \epsilon_1^{-3} \log^2 n)$ and $|W| = \bigO(k^5 \epsilon_1^{-3} \log^5 n)$ for \CFRkMed and \CFRkMeans, respectively, where $\epsilon_1$ is the guarantee parameter of the coreset. To guess the leaders, we brute-force enumerate all $k$-multisets from the coreset, which takes time $\bigO(|W|^k)$. For guessing the radius, since the distance aspect ratio of the  metric space is bounded by $n^{\bigO(1)}$, we enumerate  at most $\bigO(\epsilon_3^{-1} \log n)$ values, to consider all values such that there is a combination that is within a multiplicative factor of $(1+\epsilon_3)$ from the actual radii, for some $\epsilon_3 > 0$. 
Thus, brute-force enumeration over all $k$-multisets of radii takes  $(\bigO({\epsilon_3^{-1} \log n}))^k$  time. Combining these steps, the overall running time for Lemma~\ref{lemma:opg-fptapx} is $(\bigO(k \epsilon^{-1} \log n))^{\bigO(k)} \cdot n^{\bigO(1)}$, by choosing $\epsilon_1,\epsilon_2,\epsilon_3$ appropriately.

To recap, towards proving Theorem~\ref{thm:fptapx},  we enumerate all constraint patterns using Lemma~\ref{lemma:feasiblecp}, to identify feasible ones, and transform each of them into an instance of \OPGWCkMedDis. Then, using Lemma~\ref{lemma:opg-fptapx}, we find an approximate solution to each of them and select the solution corresponding to the instance with the minimum cost. Since optimal solution lies in one of them, this yields a $(3 + \epsilon)$- and $(9 + \epsilon)$-approximation for \CFRkMed and \CFRkMeans, respectively. For the running time, note that, we generate at most ${2}^{\bigO(tk)}$ instances \OPGWCkMedDis. The overall running time is $(\bigO(2^t k \epsilon^{-1} \log n))^{\bigO(k)} \cdot n^{\bigO(1)}$. When $t$ is constant, the running time is  $(\bigO(k \epsilon^{-1} \log n))^{\bigO(k)} \cdot n^{\bigO(1)}$.

\section{Conclusions and Discussion}
\label{sec:discussion}

In this section, we present our conclusions, outline the limitations of our work, and discuss the broader impact of our research.

\xhdr{Conclusions}
In this paper, we provide a comprehensive analysis of the computational complexity of capacitated fair-range clustering with intersecting groups, focusing on its inapproximability. Most notably, assuming \gapeth, we show that no $n^{o(k)}$-time algorithm can approximate the problem to any non-trivial factor, even when feasible solutions can be found in polynomial time, in contrast to the hardness result of~\cite{thejaswi2021diversity, thejaswi2024diversity}. This effectively make brute-force enumeration over all $k$-tuples of facilities---running in $n^{\bigO(k)}$-time---the only viable approach for approximation. 
On a positive note, we identify tractable settings where design of approximation algorithms is possible. When the number of groups $t$ is constant, we present polynomial-time $\bigO(\log k)$- and $\bigO(\log^2 k)$-approximation algorithms for the $k$-median and $k$-means objectives, respectively. We also give $(3 + \epsilon)$- and $(9 + \epsilon)$-approximation algorithms that run in $\fpt(k)$-time for these objectives. These approximation factors match the best-known guarantees for their respective unfair counterparts.

\xhdr{Limitations}
In theory, coreset construction is expected to introduce only a small $\epsilon$-factor of distortion in distances. However, in practice, building smaller-sized coresets that enable us to do brute-force enumeration requires a larger $\epsilon$, resulting to higher approximation factors. As noted in prior work~\cite{thejaswi2022clustering}, the proposed $\fpt(k)$-algorithms may not scale to large problem instances, given that the exponential factors are large. In contrast, we believe that our polynomial-time approximation algorithms are more practical and easier to implement.

\xhdr{Broader impact}
Our work focuses on fair clustering, offering theoretical insights into the challenges of designing (approximation) algorithms that support design of responsible algorithmic decision-support systems. While our algorithms provide theoretical guarantees under a formal notion of fairness, this does not automatically justify indiscriminate application. The definition of demographic groups and fairness metrics plays a crucial role in how fairness is realized in practice. Accordingly, we emphasize that our contributions are theoretical, and due caution is necessary when applying these methods in real-world settings.

\section*{Acknowledgments}
Suhas Thejaswi is supported by the European Research Council (ERC) under the European Union'{}s Horizon $2020$ research and innovation program ($945719$).

\newpage
{ 
\small
\bibliographystyle{unsrtnat}
\bibliography{references}
}

\clearpage
\newpage

\appendix

\newpage
\newpage
\section{On the Hardness of Approximation of (Capacitated) Fair-Range Clustering}
\label{app:hardness}

We present our reduction for uncapacitated \FRkMed. Since the reduction is independent of the clustering objective, the inapproximability result also applies to \FRkMeans. As capacitated variants generalize the uncapacitated case, the hardness extends to them as well.

\begin{definition}[\threesat]
An instance $(\phi,X)$ of the \threesat problem is defined on Boolean formula $\phi = C_1 \wedge C_2 \wedge \dots \wedge C_m$ consisting of $m$ clauses, where each clause $C_i =  (\ell_{i,1} \vee \ell_{i,2} \vee \ell_{i,3})$ is a disjunction of exactly three literals $\ell_{i,j} \in \{x_1, \xbar_1, \dots,x_{n'}, \xbar_{n'}\}$ over a set of variables $X=\{x_1,\dots,x_{n'}\}$. The goal is to decide whether there exists an assignment to the variables in $X$ that evaluates $\phi$ to true.
\end{definition}

We use the following hardness result for \threesat from \citet{hastad2001some} in our inapproximability proofs.

\begin{theorem}[\citet{hastad2001some}]\label{thm:3satnphard}
    For every $\epsilon>0$, it is \np-hard to decide if a given \threesat formula $\phi$ has a satisfying assignment or all assignments satisfy $<7/8+\epsilon$ fraction of clauses.
\end{theorem}

To formalize our hardness results, we define a subset of fair-range $k$-median (and $k$-means) problem---denoted \FRkMedO (and \FRkMeansO)---as instances of \FRkMed (and \FRkMeans) that admit a polynomial-time algorithm for finding feasible solutions. We establish our hardness of approximation results for this variant. Recall that a solution is feasible if it satisfies the fair-range constraints.

Next, we present polynomial-time inapproximability results for \FRkMedO in Section~\ref{ss:hard:np}, followed by parameterized inapproximability with respect to $k$ in Section~\ref{ss:hard:fpt}.

\subsection{Polynomial time inapproximability}\label{ss:hard:np}
In this subsection, we prove the following hardness result.

\nphardpoly*

Note that the trivial algorithm for \FRkMedO (or \FRkMeansO) that returns any feasible solution obtained from the oracle is a factor $\Delta$ (or $\Delta^2$) approximation, where $\Delta$ is the distance aspect ratio of the input instance.%
$^{\ref{foot:aspect-ratio}}$
%
Towards proving~\Cref{thm:hard:np1}, we show the following stronger statement that implies that this factor is essentially our best hope. For any function $g:\mathbb{N} \rightarrow \mathbb{R}_{>0}$, we denote by $g(n)$-\FRkMedO as the problem of solving \FRkMedO on instances of size $n$ with distance aspect ratio of the metric bounded by $g(n)$.

\begin{theorem}\label{thm:hard:np2}
For every
polynomial $g:\mathbb{N} \rightarrow \mathbb{R}_{>0}$ and every constant $\epsilon>0$, it is \np-hard to approximate $g(n)$-\FRkMedO to a factor  $(1-\epsilon) \frac{g(n)-2}{16}$  on tree metrics. Furthermore, for general metrics,  it is \np-hard to approximate $g(n)$-\FRkMedO to a factor $(1-\epsilon) \frac{g(n)-2}{8}$. 

In particular, the following holds assuming $\p \neq \np$. For every polynomial $g$ and for every constant $\epsilon>0$, there is no $n^{\bigO(1)}$ time algorithm that can  decide if a given instance of $g(n)$-\FRkMedO has cost at most $k$ or every feasible solution has cost $> (1-\epsilon) \frac{g(n)-2}{8}\cdot k$.
  
  Finally, there is a trivial algorithm for $g(n)$-\FRkMedO that is a factor $g(n)$-approximation.
\end{theorem}

We prove the above theorem by showing a reduction from \threesat to $g(n)$-\FRkMedO.

\subsubsection{Reduction from \threesat to $g(n)$-\FRkMedO}
Here we show the following result.
\begin{theorem}\label{thm:redthreesat}
    Given a polynomial $g:\mathbb{N} \rightarrow \mathbb{R}_{>0}$, a constant $\epsilon>0$,  and an instance  $\phi$ of \threesat on $n'$ variable set $X=\{X_1,\dots, X_{n'}\}$ with $m$ clauses $C = \{C_1,\dots, C_m\}$, there is a $(m\,n')^{\bigO(1)}$ time algorithm that computes an instance $\Izcal$ of $g(n)$-\FRkMedO of size $n$ such that the following holds.
    \begin{enumerate}
        \item Parameters: $|C|=m, |F|=8m, n=9m+1, k=m, t\le 2n'm^2+m$
        \item (Yes case) If there is a satisfying assignment $\sigma$ to $X$ that satisfies all the clauses of $\phi$, then there is a feasible solution $S_\sigma$ to $\Izcal$ that has cost at most $k$
        \item (No case) If every assignment satisfies  $<(7/8+\epsilon)$ fraction of clauses, then every feasible solution to $\Izcal$ has cost $>(1-8\epsilon)\frac{g(n)-2}{16}k$.
    \end{enumerate}
\end{theorem}

\begin{proof}
Let $(\phi,X)$ be the given instance of \threesat. We construct an instance $\Izcal$ of \FRkMedO using $(\phi,X)$ as follows. 

\xhdr{Construction}
Let $D\ge 1$ be a fixed number.
For every clause $C_i$ of $\phi$, we create a client $c_i$, and create $8$ facilities $f^1_i,\dots,f^8_i$ in $\Izcal$, corresponding to the partial assignments to the variables of $C_i$. 
Next, we add a dummy node $s$.
Now, we construct a metric over $C \cup F \cup \{s\}$ of $\Izcal$ as follows. First, we create  a weighted bipartite graph with left partition $C$ and right partition $F$. For each $i \in [m]$, add edges between $c_i$ and $f^j_i$, for $j \in [8]$. Furthermore, if the partial assignment corresponding to $f^j_i$ satisfies the clause $C_i$, corresponding to $c_i$, assign the weight of the edge to be $1$, otherwise assign the weight to be $D$. Finally, add unit weight edges from $s$ to all clients.
Next, we create the groups in $\GG$ as follows. Specifically, we create two types of groups -- \emph{clause groups} and \emph{assignment groups}. For every clause $C_i$ of $\phi$, create a clause group $G_i$ that contain all the facilities $f^1_i,\dots, f^8_i$, and set $\alpha_i=\beta_i=1$.\footnote{In fact, we can set $\alpha_i=1$ and $\beta_i=m$, which captures the lower bound setting of~\cite{thejaswi2021diversity}. See Section~\ref{ss:lowerbound:hard} more details.}
Let $G_C=\{G_1,\dots,G_m\}$ be the set of clause groups.
Next, for every variable $X_j$, for every assignment $a \in \{0,1\}$, and for every pair of clauses $C_i$ and $C_{i'}$ containing $X_j$,  we create an assignment group $G_{X_j \mapsto a}^{(C_i,C_{i'})}$ that contains all facilities in $G_i \in G_C$ that assign $a$ to $X_j$ and all facilities in $G_{i'} \in G_C$ that assign $1-a$ to $X_j$.
We set the corresponding requirements as $\alpha_{X_j \mapsto a}^{(C_i,C_{i'})} =  \beta_{X_j \mapsto a}^{(C_i,C_{i'})}=1$.
Next, we set $k=m$. This completes the construction.
See Section~\ref{sec:hardness} for a pictorial depiction of the reduction.
\\

We first verify the parameters of the instance $\Izcal$. Since, we create client for every clause, we have $|C|=m$. For every clause, we create $8$ facilities, and hence $|F|=8m$. Additionally, we have a dummy node in the metric space, implying $n=9m+1$. Finally, we create $m$ clause groups, and at most $2n'm^2$ assignment groups, and hence $t \le m+2n'm^2$, as desired. Now, we claim that $\Izcal$ is an instance of \FRkMedO.
\begin{claim}\label{cl:assgnsol}
    For every assignment $\sigma: X \rightarrow \{0,1\}$, there is a feasible solution $S_\sigma$ to $\Izcal$. Therefore, $\Izcal$ is an instance of \FRkMedO.
\end{claim}
\begin{proof}
     Consider the solution $S_\sigma \subseteq F$ of size $k$ that contains, for every $i \in [k] = [m]$, a facility $f^j_i\in G_i, j \in [8]$ such that $f^j_i$ corresponds to the partial assignment on the variables of $C_i$ due to $\sigma$. First, we claim that $S_\sigma$ is a feasible solution to $\Izcal$. To see this, note that, for every clause group $G_i \in G_C$, we have $|S_\sigma \cap G_i|=1$, by construction. Furthermore, for a variable $X_j$, and clauses $C_i$ and $C_{i'}$ containing $X_j$, we claim that $|S_\sigma \cap G_{X_j \mapsto \sigma(X_j)}^{(C_i,C_{i'})}|=1$. To see this, let $G_{i, X_j \mapsto \sigma(X_j)} \subseteq G_i$ be the set of facilities of $G_i$ that correspond to the partial assignments to the variables of $C_i$ that assign $\sigma(X_j)$ to $X_j$. Similarly, let  $G_{i', X_j \mapsto 1-\sigma(X_j)} \subseteq G_i$ be the set of facilities of $G_i$ that correspond to the partial assignments to the variables of $C_i$ that assign $1-\sigma(X_j)$ to $X_j$. Then, note that $G_{X_j \mapsto \sigma(X_j)}^{(C_i,C_j)} = G_{i, X_j \mapsto \sigma(X_j)} \cup G_{i', X_j \mapsto 1-\sigma(X_j)}$. Now, let
$f_i\in G_i$ and $f_{i'} \in G_{i'}$ be the facility corresponding to the partial assignment to the variable of $C_i$ and $C_{i'}$, respectively, due to $\sigma$. 
First note that $f_i,f_{i'} \in S_\sigma$.
Next we have, $f_i \in G_{i, X_j \mapsto \sigma(X_j)}$, hence  $f_i \in G_{X_j \mapsto \sigma(X_j)}^{(C_i,C_j)}$. Now, observe that $|S_\sigma \cap  G_{i, X_j \mapsto \sigma(X_j)}|=1$, since $G_{i, X_j \mapsto \sigma(X_j)} \subseteq G_i$ and $|S_\sigma \cap G_i|=1$ by construction.  
However, since $|S_\sigma \cap G_{i'}|=1$ and since $G_{i', X_j \mapsto \sigma(X_j)} \subseteq G_{i'}$, we have that 
$|S_\sigma \cap  G_{i', X_j \mapsto 1-\sigma(X_j)}|=0$, as  $f_{i'} \in S_\sigma$ but $f_{i'} \notin  G_{i', X_j \mapsto 1-\sigma(X_j)}$.
Therefore, $|S_\sigma \cap G_{X_j \mapsto \sigma(X_j)}^{({C}_i,{C}_{i'})}|=1$, and hence $S_\sigma$ is a feasible solution to $\Izcal$.

\end{proof}
\begin{lemma}[Yes case]\label{lm:hard2:yes}
    If $\phi$ has a satisfying assignment, then there exists a feasible solution to $\Izcal$ with cost at most $k$.
\end{lemma}
\begin{proof}
  Suppose there is an assignment $\sigma: X \rightarrow \{0,1\}$ to $X$ such that $\phi$ is satisfiable. Then, consider the solution $S_\sigma \subseteq F$ of size $k$ obtained from Claim~\ref{cl:assgnsol}.
As $S_\sigma$ is a feasible solution, we have $|S_\sigma \cap G_i|=1$, for all $i \in [m]$. Therefore, let $f_i \in G_i$ be the facility that was picked from $G_i$ during the construction of $S_\sigma$. Since $\sigma$ satisfies clause $C_i$ of $\phi$, we have that the weight of the edge between $f_i$  and $c_i$ is $1$. Hence, the cost of $S_\sigma$ is $m=k$.  
\end{proof}

\begin{lemma}[No case]\label{lm:hard2:no}
    If every assignment to $\phi$ satisfies at most  $(7/8+\epsilon)$ fraction of clauses, then every feasible solution to $\Izcal$ has cost more than $\frac{(1-8\epsilon)D\cdot k}{8}$.
\end{lemma}
\begin{proof}
We will prove the contrapositive of the statement.
Suppose there is a feasible solution $S$ to $\Izcal$ of size $k$ with cost at most  $\frac{(1-8\epsilon)D\cdot k}{8}$. We will show an assignment $\sigma_S : X \rightarrow \{0,1\}$ to the variables of $\phi$ such that $\sigma_S$ satisfies at least $(7/8+\epsilon)$ fraction of clauses.
Since $S$ satisfies the diversity constraints on the clause groups, we have that $|S \cap G_i|=1$, for every $i \in [m]$. Let $f_i \in G_i$ be the facility in $S$, for the clause group $G_i$. Let $\sigma_i: X\vert_{C_i} \rightarrow \{0,1\}$ be the partial assignment to the variables of $C_i$ corresponding to facility $f_i \in G_i \cap S$. We claim that the partial assignments $\{\sigma_i\}_{i \in [m]}$ are consistent, i.e., there is no variable $X_j \in X$ that receives different assignments from $\sigma_{i}, \sigma_{i'}$, for some $i,i' \in [m]$. Suppose, for the contradiction, there exist $X_j \in X$, and $i,i' \in [m]$  such that $\sigma_{i}(X_j) \neq \sigma_{i'}(X_j)$. Without loss of generality, assume that  $\sigma_{i}(X_j)=1$ and  $\sigma_{i'}(X_j)=0$, and consider the assignment group $G^{(C_i,C_{i'})}_{X_j \mapsto 1}$. Then note that both $f_i, f_{i'} \in G^{(C_i,C_{i'})}_{X_j \mapsto 1}$, since $f_i \in G_i$ corresponds to $\sigma_i$ such that $\sigma_i(X_j)=1$, whereas  $f_{i'} \in G_{i'}$ corresponds to $\sigma_{i'}$ such that $\sigma_{i'}(X_j)=0$. Hence, $|S \cap G^{(C_i,C_{i'})}_{X_j \mapsto 1}|=2 > \beta^{(C_i, C_{i'})}_{X_j \mapsto 1}=1$, contradicting the fact that $S$ is a feasible solution to $\Izcal$. Therefore, the partial assignments $\{\sigma_i\}_{i \in [m]}$ are consistent. Now consider the global assignment $\sigma_S: X \rightarrow \{0,1\}$ obtained from these partial assignments.\footnote{If a variable  is not assigned by any partial assignment, we assign it an arbitrary value from $\{0,1\}$.\label{foot:unassgnvar}} The following claim says that $\sigma_S$ satisfies at least $(\frac{7}{8}+\epsilon)m$ clauses, contracting Theorem~\ref{thm:3satnphard}.
\begin{claim}\label{cl:hard2:nocase}
    $\sigma_S$ satisfies at least $(\frac{7}{8}+\epsilon)m$ clauses of $\phi$.
\end{claim}
\begin{proof}
Let $C' \subseteq C$ be the clients that have a center in $S$ at a distance $1$, and let $|C'|=\ell$. Note that the closest center for client  $c_i \in C \setminus C'$ in $S$ is at a distance $D$ since $|S \cap G_i|$=1, and the distance between $c_i$ and facilities in $G_i$ is either $1$ or $D$. Therefore, the cost of $S$ is $\sum_{c \in C'} d(c,S) + \sum_{c \in C \setminus C'} d(c,S)        = \ell\cdot 1+ (m-\ell)D$. Since, we assumed that the cost of $S$ is at most  $\frac{(1-8\epsilon)D\cdot k}{8}$, we have that $|C'| =\ell > (\frac{7}{8}+\epsilon)m$. This means that for at least $(\frac{7}{8}+\epsilon)m$ clients in $\Izcal$, there is a facility in $S$ at a distance $1$. Hence, for every such client $c_i \in C'$, the corresponding clause $C_i$ is  satisfied by the partial assignment $\sigma_i$, implying that the number of clauses satisfied by $\sigma_S$ is at least $(\frac{7}{8}+\epsilon)m$.
\end{proof}
This finishes the proof the lemma.
\end{proof}

We finish the proof of the theorem by using $g(n)=2D+2$.

\end{proof}

\subsubsection{Proofs}
\xhdr{Proof of Theorem~\ref{thm:hard:np2}}
Fix $g(n)$ and $\epsilon>0$. Let $\phi$ be an instance of \threesat obtained from  Theorem~\ref{thm:3satnphard}. Using the construction of Theorem~\ref{thm:redthreesat} on $g(n), \epsilon/8$, and $\phi$, we obtain an instance $\Izcal$ of $g(n)$-\FRkMedO in $(|\phi|^{\bigO(1)})$ such that 
\begin{itemize}
    \item If $\phi$ has a satisfying assignment, then there exists a feasible solution to $\Izcal$ with cost $k$
    \item If every assignment to $\phi$ satisfies $<(\frac{7+\epsilon}{8})$ fraction of clauses, then every feasible solution to $\Izcal$ has cost $> (1-\epsilon) \frac{g(n)-2}{16}k$.
\end{itemize}
Therefore, it is \np-hard to decide if a given instance of $g(n)$-\FRkMedO has cost at most $k$ or $> (1-\epsilon) \frac{g(n)-2}{16}k$. Finally, observe that $\Izcal$ is defined on a tree metric (in fact, a depth $2$ rooted tree with root $s$).

For general metrics, we obtain slightly better constants in the lower bound. The idea is that, in the construction of Theorem~\ref{thm:redthreesat}, instead of adding the dummy vertex $S$, we add the missing edges on the graph on $C \cup F$ with weight $D$. Therefore, the distance aspect ratio $g(n)=D$, and hence the bound follows from Lemma~\ref{lm:hard2:no}.

Finally, for the upper bound, consider an instance $\Izcal$ of $g(n)$-\FRkMedO of size $n$, for some polynomial $g(n)$ with $\opt$ as the optimal cost. 
Let $d_{\max}$ and $d_{\min}$ be the largest and the smallest distances in the metric $d$
of $\Izcal$, respectively. Note that $g(n)=\frac{d_{\max}}{d_{\min}}$.
Let $S$ be a feasible solution obtained from the oracle $\mathcal{O}$ in $(n)^\bigO(1)$ time. Then, we claim that $S$ is a factor $g(n)$-approximate solution to $\Izcal$. This follows since, $S$ is a feasible solution with cost
\[
=\sum_{c \in C} d(c,S) \le n \cdot d_{\max}\le g(n)\cdot \opt,
\]
since $\opt \ge n \cdot d_{min}$. 

\xhdr{Proof of Theorem~\ref{thm:hard:np1}}
Suppose there is an algorithm $\mathcal{A}$ that, given an instance $\Izcal$ of \FRkMedO of size $n$ outputs a feasible solution $S_\Izcal$ in $(n)^{\bigO(1)}$ time with cost at most $p(n) \cdot \opt(\Izcal)$,  for some polynomial function $p$, where $\opt(\Izcal)$ is the optimal cost of $\Izcal$. 
Let $\Izcal$ be the hard instance of $g(n)$-\FRkMedO on tree metric obtained from Theorem~\ref{thm:hard:np1} with $g(n)=32 p(n)+2$ and $\epsilon=1/2$.
We will use $(n)^{\bigO(1)}$-time algorithm $\mathcal{A}$ to construct a polynomial time algorithm $\mathcal{B}$ that decides if  (Yes) there is a feasible solution to $\Izcal$ with cost at most $k$ or (No) every feasible solution to $\Izcal$ has cost $> (1-\epsilon) \frac{g(n)-2}{16}\cdot k = \frac{g(n)-2}{32}\cdot k$, contradicting the assumption $\p \neq \np$. 
Our algorithm $\mathcal{B}$ first computes a feasible solution $S$ to $\Izcal$ using $\mathcal{A}$. Then,  $\mathcal{B}$ says Yes if the cost of $S$ is at most  $p(n)\cdot k$, and No otherwise.
To see that $\mathcal{B}$ correctly decides on $\Izcal$, note that the cost of $S$ is at most $p(n)\cdot \opt(\Izcal)$. If $\opt(\Izcal) \le k$, then the cost of $S$ is at most $p(n)\cdot k =(\frac{g(n)}{32}-2)\cdot k$, while if $\opt(\Izcal)> (1-\epsilon) \frac{g(n)-2}{16}\cdot k = \frac{g(n)-2}{32}\cdot k=p(n)\cdot k$, finishing the proof.

\subsection{Parameterized inapproximability}\label{ss:hard:fpt}
In this section, we strengthen Theorem~\ref{thm:hard:np1} to rule out $n^{o(k)}$ time algorithms for obtaining the corresponding guarantee. Our lower bound is based on the following assumption, called \gapeth.

\begin{hypothesis}[(Randomized) Gap Exponential Time Hypothesis (\gapeth)~\cite{dinur2016mildly,manurangsi2017abirthday}]\label{hyp: gapeth}
    There exists constants $\epsilon, \tau >0$ such that no randomized algorithm when given an instance $\phi$ of \threesat on $n'$ variables and $O(n')$ clauses  can distinguish the following cases correctly with probability $2/3$ in time $O(2^{\tau n'})$.
    \begin{itemize}
        \item there exists an assignment for $\phi$ that satisfies all the clauses
        \item every assignment satisfies $< (1-\epsilon)$ fraction of clauses in $\phi$.
    \end{itemize}
\end{hypothesis}

In particular, \gapeth implies the following statement.
\begin{theorem}\label{thm:gapeth}
    Assuming \gapeth, there exist $\epsilon>0$ such that there is no $2^{o(n')}$-time algorithm that given an instance $\phi$ of \threesat on $n'$ variables and $O(n')$ clauses can decide correctly with probability $2/3$ if there exists an assignment for $\phi$ that satisfies all the clauses or every assignment satisfies $< (1-\epsilon)$ fraction of clauses in $\phi$.
\end{theorem}
We show the following hardness results.

\enuhard*

Towards proving this, we show the following reduction which is similar to Theorem~\ref{thm:redthreesat}.

\begin{theorem}\label{thm:redgapthreesat}
    Given a polynomial $p:\mathbb{N} \rightarrow \mathbb{R}_{\ge 1}$, a constant $1\ge \epsilon>0$, an integer $\kappa \in \mathbb{Z}_{+}$, and an instance  $\phi$ of \threesat on $n'$ variable set $X=\{X_1,\dots, X_{n'}\}$ with $m=O(n')$ clauses $C = \{C_1,\dots, C_m\}$, there is a $2^{O(n'/k)}(n')^{\bigO(1)}$-time algorithm that computes an instance $\Izcal$ of $g(n)$-\FRkMedO of size $n$, such that the following holds.
    \begin{enumerate}
        \item Parameters: $|C|=\kappa, |F|=\kappa \cdot 2^{O(n'/\kappa)}, n=k(2^{O(n'/\kappa)}+1)+1, k=\kappa, t\le \kappa+2n'\kappa^2$
        \item (Yes case) If there is a satisfying assignment $\sigma$ to $X$ that satisfies all the clauses of $\phi$, then there is a feasible solution $S_\sigma$ to $\Izcal$ that has cost at most $k$
        \item (No case) If every assignment satisfies  $<(1-\epsilon)$ fraction of clauses, then every feasible solution to $\Izcal$ has cost $> p(n)\cdot k$.
    \end{enumerate}
\end{theorem}
\begin{proof}
Let $(\phi,X)$ be the given instance of \threesat. We construct an instance $\Izcal$ of \FRkMedO using $(\phi,X)$ as follows. 

\xhdr{Construction}
Without loss of generality, we assume that $k$ divides $m$, and let  $m/k=\ell\in \mathbb{Z}_+$. We start by partitioning the clauses $C_1,\dots,C_m$ of $\phi$ into $\kappa$ parts $\tilde{C}_1,\dots, \tilde{C}_\kappa$ arbitrarily, where each part contains $\ell=m/k$ clauses of $\phi$. We call each part $\tilde{C}_i, i\in [\ell]$, a \emph{super clause}.
Let $D \ge 1$ be a fixed real, which will be decided later. 
Let $\rho_i$ be the number of partial assignments to the variables of clauses in $\tilde{C}_i$. Then, note that $\rho_i \le 2^{O(m/\kappa)}=2^{O(\ell)}$, since $\tilde{C}_i$ contains $\ell$ clauses of $\phi$.
For every super clause $\tilde{C}_i$ of $\phi$, we create a client $c_i$, and create ${\rho_i}=2^{O(m/k)}$ facilities $f^1_i,\dots,f^{\rho_i}_i$ in $\Izcal$, corresponding to the partial assignments to the variables of clauses in $\tilde{C}_i$. Next, we add a dummy node $s$.
Now, we construct a metric over $C \cup F \cup \{s\}$ of $\Izcal$ as follows. First, we create a weighted bipartite graph with left partition $C$ and right partition $F$. For each $i \in [\kappa]$, add edges between $c_i$ and $f^j_i$, for $j \in [\rho_i]$. Furthermore, if the partial assignment corresponding to $f^j_i$ satisfies all the clauses in $\tilde{C}_i$, corresponding to $c_i$, assign the weight of the edge to be $1$, otherwise assign the weight to be $D$. Finally, add unit weight edges from $s$ to all clients.
Next, we create the groups in $\GG$ as follows. Specifically, we create two types of groups -- \emph{\sclause groups} and \emph{assignment groups}. For every \sclause $\tilde{C}_i$ of $\phi$, create a \sclause group $G_i$ that contain all the facilities $f^1_i,\dots, f^{\rho_i}_i$, and set $\alpha_i=\beta_i=1$.\footnote{In fact, we can set $\alpha_i=1$ and $\beta_i=m$, which captures the lower bound setting of~\cite{thejaswi2021diversity}. See Section~\ref{ss:lowerbound:hard} more details.}
Let $G_{\tilde{C}}=\{G_1,\dots,G_\kappa\}$ be the set of \sclause groups.
For a \sclause $\tilde{C}_i$, we say that $\tilde{C}_i$ \emph{contains} variable $X_j \in X$ if there is a clause in $\tilde{C}_i$ that contains $X_j$.
Now, for every variable $X_j$, for every assignment $a \in \{0,1\}$, and for every pair of \sclauses $\tilde{C}_i$ and $\tilde{C}_{i'}$ that both contain $X_j$,  we create an assignment group $G_{X_j \mapsto a}^{(\tilde{C}_i,\tilde{C}_{i'})}$ that contains all facilities in $G_{i} \in G_{\tilde{C}}$ that assign $a$ to $X_j$ and all facilities in $G_{{i'}} \in G_{\tilde{C}}$ that assign $1-a$ to $X_j$.
We set the corresponding requirements as $\alpha_{X_j \mapsto a}^{(\tilde{C}_i,\tilde{C}_{i'})} =  \beta_{X_j \mapsto a}^{(\tilde{C}_i,\tilde{C}_{i'})}=1$.
Finally, we set $k=\kappa$. This completes the construction.\\

We first verify the parameters of the instance $\Izcal$. Since, we create a client for every \sclause, we have $|C|=\kappa=k$. For every \sclause $\tilde{C}_i$, we create $\rho_i$ facilities, and hence $|F|=\sum_{i \in [\kappa]} \rho_i = k\cdot 2^{O(m/\kappa)}$. Additionally, we have a dummy node in the metric space, implying $n=\kappa\cdot (2^{O(m/k)}+1)+1$. Finally, we create $\kappa$ clause groups, and at most $2n'\kappa^2$ assignment groups, and hence $t \le \kappa+2n'\kappa^2$, as desired. Now, we claim that $\Izcal$ is an instance of \FRkMedO.
\begin{claim}\label{cl:assgnsoleth}
    For every assignment $\sigma: X \rightarrow \{0,1\}$, there is a feasible solution $S_\sigma$ to $\Izcal$. Therefore, $\Izcal$ is an instance of \FRkMedO.
\end{claim}
\begin{proof}
     Consider the solution $S_\sigma \subseteq F$ of size $k$ that contains, for every $i \in [k]$, a facility $f^j_i\in G_i, j \in [\rho_i]$ such that $f^j_i$ corresponds to the partial assignment on the variables of the clauses in $\tilde{C}_i$ due to $\sigma$. We claim that $S_\sigma$ is a feasible solution to $\Izcal$. To see this, note that, for every \sclause group $G_i \in G_{\tilde{C}}$, we have $|S_\sigma \cap G_i|=1$, by construction. Furthermore, for a variable $X_j$, and \sclauses $\tilde{C}_i$ and $\tilde{C}_{i'}$ containing $X_j$, we claim that $|S_\sigma \cap G_{X_j \mapsto \sigma(X_j)}^{(\tilde{C}_i,\tilde{C}_{i'})}|=1$. To see this, let $G_{i, X_j \mapsto \sigma(X_j)} \subseteq G_i$ be the set of facilities of $G_i$ that correspond to the partial assignments to the variables of $\tilde{C}_i$ that assign $\sigma(X_j)$ to $X_j$. Similarly, let  $G_{i', X_j \mapsto 1-\sigma(X_j)} \subseteq G_i$ be the set of facilities of $G_i$ that correspond to the partial assignments to the variables of $\tilde{C}_i$ that assign $1-\sigma(X_j)$ to $X_j$. Then, note that $G_{X_j \mapsto \sigma(X_j)}^{(\tilde{C}_i,\tilde{C}_j)} = G_{i, X_j \mapsto \sigma(X_j)} \cup G_{i', X_j \mapsto 1-\sigma(X_j)}$. Now, let
$f_i\in G_i$ and $f_{i'} \in G_{i'}$ be the facilities corresponding to the partial assignment to the variables of clauses of $\tilde{C}_i$ and $\tilde{C}_{i'}$, respectively, due to $\sigma$. 
First note that $f_i,f_{i'} \in S_\sigma$.
Next we have, $f_i \in G_{i, X_j \mapsto \sigma(X_j)}$, hence  $f_i \in G_{X_j \mapsto \sigma(X_j)}^{(\tilde{C}_i,\tilde{C}_j)}$. Now, observe that $|S_\sigma \cap  G_{i, X_j \mapsto \sigma(X_j)}|=1$, since $G_{i, X_j \mapsto \sigma(X_j)} \subseteq G_i$ and $|S_\sigma \cap G_i|=1$ by construction.  
However, since $|S_\sigma \cap G_{i'}|=1$ and since $G_{i', X_j \mapsto \sigma(X_j)} \subseteq G_{i'}$, we have that 
$|S_\sigma \cap  G_{i', X_j \mapsto 1-\sigma(X_j)}|=0$, as  $f_{i'} \in S_\sigma$ but $f_{i'} \notin  G_{i', X_j \mapsto 1-\sigma(X_j)}$.
Therefore, $|S_\sigma \cap G_{X_j \mapsto \sigma(X_j)}^{(\tilde{C}_i,\tilde{C}_{i'})}|=1$, and hence $S_\sigma$ is a feasible solution to $\Izcal$.
\end{proof}
\begin{lemma}[Yes case]\label{lm:hard2:yes}
    If $\phi$ has a satisfying assignment, then there exists a feasible solution to $\Izcal$ with cost at most $k$.
\end{lemma}
\begin{proof}
  Suppose there is an assignment $\sigma: X \rightarrow \{0,1\}$ to $X$ such that $\phi$ is satisfiable. Then, consider the solution $S_\sigma \subseteq F$ of size $k$ obtained from Claim~\ref{cl:assgnsoleth}.
  
As $S_\sigma$ is a feasible solution, we have $|S_\sigma \cap G_i|=1$, for all $i \in [m]$. Therefore, let $f_i \in G_i$ be the facility that was picked from $G_i$ during the construction of $S_\sigma$. Since $\sigma$ satisfies all the clauses in $\tilde{C}_i$ of $\phi$, we have that the weight of the edge between $f_i$  and $c_i$ is $1$. Hence, the cost of $S_\sigma$ is $k$.  
\end{proof}

\begin{lemma}[No case]\label{lm:hard2:no}
    If every assignment to $\phi$ satisfies at most  $(1-\epsilon)$ fraction of clauses, then every feasible solution to $\Izcal$ has cost more than $(1+\epsilon\cdot (D-1))k$.
\end{lemma}
\begin{proof}
We will prove the contrapositive of the statement.
Suppose there is a feasible solution $S$ to $\Izcal$ of size $k$ with cost at most  $(1+\epsilon\cdot (D-1))k$. We will show an assignment $\sigma_S : X \rightarrow \{0,1\}$ to the variables of $\phi$ such that $\sigma_S$ satisfies at least $(1-\epsilon)$ fraction of clauses.
Since $S$ satisfies the diversity constraints on the \sclause groups, we have that $|S \cap G_i|=1$, for every $i \in [k]$. Let $f_i \in G_i$ be the facility in $S$, for the \sclause group $G_i$. Let $\sigma_i: X\vert_{\tilde{C}_i} \rightarrow \{0,1\}$ be the partial assignment to the variables of the clauses in $\tilde{C}_i$ corresponding to facility $f_i \in G_i \cap S$. We claim that the partial assignments $\{\sigma_i\}_{i \in [k]}$ are consistent, i.e., there is no variable $X_j \in X$ that receives different assignments from $\sigma_{i}, \sigma_{i'}$, for some $i,i' \in [k]$. Suppose, for the contradiction, there exist $X_j \in X$, and $i \neq i' \in [k]$  such that $\sigma_{i}(X_j) \neq \sigma_{i'}(X_j)$. Without loss of generality, assume that  $\sigma_{i}(X_j)=1$ and  $\sigma_{i'}(X_j)=0$ and consider the assignment group $G^{(\tilde{C}_i,\tilde{C}_{i'})}_{X_j \mapsto 1}$. Then, note that both $f_i, f_{i'} \in G^{(\tilde{C}_i,\tilde{C}_{i'})}_{X_j \mapsto 1}$, since $f_i \in G_i$ corresponds to $\sigma_i$ such that $\sigma_i(X_j)=1$, whereas  $f_{i'} \in G_{i'}$ corresponds to $\sigma_{i'}$ such that $\sigma_{i'}(X_j)=0$. Hence, $|S \cap G^{(\tilde{C}_i,\tilde{C}_{i'})}_{X_j \mapsto 1}|\geq 2 > 1= \beta^{(\tilde{C}_i, \tilde{C}_{i'})}_{X_j \mapsto 1}$, contradicting the fact that $S$ is a feasible solution to $\Izcal$. Therefore, the partial assignments $\{\sigma_i\}_{i \in [k]}$ are consistent. Now consider the global assignment $\sigma_S: X \rightarrow \{0,1\}$ obtained from these partial assignments.$^{\ref{foot:unassgnvar}}$  The following claim says that $\sigma_S$ satisfies at least $(1-\epsilon)m$ clauses, contracting Theorem~\ref{thm:3satnphard}.
\begin{claim}\label{cl:hard2:nocase}
    $\sigma_S$ satisfies at least $(1-\epsilon)m$ clauses of $\phi$.
\end{claim}
\begin{proof}
    Let $C' \subseteq C$ be the clients that have a center in $S$ at a distance $1$, and let $|C'|=\mu$. Note that the closest center for client  $c_i \in C \setminus C'$ in $S$ is at a distance $D$. This is due to the fact that $|S \cap G_i|$=1, and the closest center to $c_i$ in $S_\sigma$ is in $G_i$, which is at a distance $D$ from $c_i$. Therefore, the cost of $S$ is $\sum_{c \in C'} d(c,S) + \sum_{c \in C \setminus C'} d(c,S)        = \mu\cdot 1+ (k-\mu)\cdot D$. Since, we assumed that the cost of $S$ is at most  $(1+\epsilon\cdot (D-1))k$, we have that $|C'| =\mu > (1-\epsilon)k$.  
    This means that for more than $(1-\epsilon)k$ clients in $\Izcal$, there is a facility in $S$ at a distance $1$. We call such a client, a \emph{good} client for $S$, the corresponding \sclause $\tilde{C}_i$ a \emph{good} \sclause for $\sigma_S$.   
    Now, note that a \sclause contains $\ell=m/k$ clauses of $\phi$. 
    Hence, all the clauses in a good \sclause $\tilde{C}_i$ for $\sigma_S$, corresponding to a good client $c_i\in C'$ for $S$, are satisfied by the partial assignment $\sigma_i$, implying that the number of clauses satisfied by $\sigma_S$ is at least $(1-\epsilon)m$.
\end{proof}
This finishes the proof the lemma.
\end{proof}

We finish the proof of the theorem by using $D=\frac{p(n)}{\epsilon}$.

\end{proof}

\xhdr{Proof of Theorem~\ref{thm:hard:enu1}}
Suppose there is an algorithm $\cal{A}$ that, given an instance $\Izcal$ \FRkMedO of size $n$ on tree metric, runs in time $f(k)n^{O(k/h(k))}$, for some non-decreasing and unbounded functions $f$ and $h$, and produces a feasible solution with cost at most $p(n)\cdot \opt(\Izcal)$, for some polynomial $p$. Then, using $\cal{A}$, we will design an algorithm $\cal{B}$ that, given an instance $\phi$ of \threesat on $n'$ variables and $O(n')$ clauses and any $\epsilon'>0$, correctly decides if $\phi$ has a satisfying assignment or every assignment satisfies $< (1-\epsilon')$ fraction of clauses in $\phi$, and runs in time $2^{o(n')}$, contradicting Theorem~\ref{thm:gapeth}.

Given a \threesat formula $\phi$ on $n'$ variables and $O(n')$ clauses, and $1 \ge \epsilon>0$, the algorithm $\cal{B}$ does the following, using the algorithm $\cal{A}$ for \FRkMedO that runs in time $f(k)n^{O(k/h(k))}$. Without loss of generality, we assume that $f(k)\ge 2^k$. Given $\phi$, let $\kappa$ be the largest integer such that $f(\kappa)\le n'$.  Note that, the value of $\kappa$ thus  computed depends on $n'$, and hence let $\kappa =q(n')$, for some non-decreasing and unbounded function $q$.
Since, $f(\kappa)\ge 2^{\kappa} \ge \kappa$, and $f(\kappa)\le n'$, we have that $\kappa \le 2\log n'$ and $q(n') \le n'$.
Given $\phi$ and  $\epsilon$, $\cal{B}$ first uses the reduction of Theorem~\ref{thm:redgapthreesat} on $\kappa=q(n')$ and polynomial $p$, to obtain an instance $\Izcal$ of \FRkMedO. Next, $\cal{B}$ runs algorithm $\cal{A}$ to obtain a feasible solution $S$ with cost  $p(n)\cdot \opt(\Izcal)$, returns Yes if cost of $S$ is at most $p(n)\cdot k$, and No otherwise. To see the correctness of $\cal{B}$ on $\phi$, note that if $\phi$ has a satisfying assignment then $\Izcal$ has a feasible solution with cost $k=\kappa$. In this case, the cost of $S$ is at most $p(n)\cdot k$. On the other hand, if every assignment satisfied $<(1-\epsilon)$ fraction of clauses of $\phi$, then every feasible solution to $\Izcal$ has cost  $> p(n)\cdot k$, and hence the cost of $S$ is  $> p(n)\cdot k$. Therefore, $\cal{B}$ correctly decides if $\phi$ has a satisfying assignment or every assignment satisfies $< (1-\epsilon')$ fraction of clauses. Finally, the running time of $\cal{B}$ is bounded by
\[
O(f(k)n^{O(k/h(k))}) \le O(n' \cdot (2k\cdot 2^{O(n'/k)})^{\frac{k}{h(k)}}) \le O(n'\cdot 2^{o(n')}),
\]
as desired, since $k=\kappa \le 2\log n'$.

\subsection{Hardness for the lower-bound only \FRkMedO (and \FRkMeansO)}\label{ss:lowerbound:hard}
In this section, we sketch the changes required in the reductions of Theorem~\ref{thm:redthreesat} and Theorem~\ref{thm:redgapthreesat}, such that these reductions construct instances with lower-bound only requirements.
Note that, in both the constructions, we create two types of groups: \emph{clause groups} and \emph{assignment groups}, and make the lower and upper bound requirements for every group to be $1$. Therefore, without particularly fixing any construction, we focus on showing how to transform the group requirements to have lower bound only requirements such that  Theorem~\ref{thm:redthreesat} and \ref{thm:redgapthreesat} remain true.

As mentioned before, every group $G_i$ constructed by Theorem~\ref{thm:redthreesat} and \ref{thm:redgapthreesat} for the instance $\cI$ of \FRkMedO has $\alpha_i=\beta_i=1$. For the lower-bound only construction, we simply drop the $\beta_i$'s. Hence, $G_i$ has a lower-bound only requirement $\alpha_i=1$.  Let us denote the obtained instance as $\cI_L$, to denote the fact that it is obtained from $\cI$ by keeping only the lower-bounds.
To argue that Theorem~\ref{thm:redthreesat} and Theorem~\ref{thm:redgapthreesat} remain true, we claim that a solution $S$ is feasible to $\cI$ if and only if $S$ is feasible to $\cI_L$.  Since we only relaxed the requirements for $\cI_L$, any feasible solution to $\cI$ is also a feasible solution to $\cI_L$, proving the forward direction. Now, for the reverse direction, consider a feasible solution $S$ to $\cI_L$. We claim that, for every group $G_i$ of $\cI$, it holds that $|S\cap G_i| =1$, and hence $S$ satisfies the upper bound and lower bound requirements for $G_i$, since $\alpha_i=\beta_i=1$. Therefore, $S$ is a feasible solution to $\cI$, finishing the claim. Towards this, suppose $G_i$ is a clause group. Then, observe that there are $k$ clause groups in $\cI_L$ which are mutually disjoint with each other. Since $|S| \leq k$, and each clause group has a lower-bound requirement of $1$, it holds that $|S \cap G_i|=1$, as required for $\cI$. Finally, suppose $G_i$ is an assignment group. Without loss of generality, suppose $G_i$ corresponds to the assignment group $G^{(C_i,C_{i'})}_{X_j \mapsto a}$, for some pair of (super) clauses $C_i$ and $C_{i'}$, and $a \in \{0,1\}$. Suppose for the contradiction, we have that $|S \cap G^{(C_i,C_{i'})}_{X_j \mapsto a}| = 2$. Recall that $G^{(C_i,C_{i'})}_{X_j \mapsto a}$ contains all the facilities of clause group $G_i$ that assign $a$ to $X_j$ and all facilities in clause group $G_{i'}$ that assign $1-a$ to $X_j$. Let us denote by $G^{C_i}_{X_j \mapsto a} \subseteq G_i$, and $G^{C_{i'}}_{X_j \mapsto 1-a} \subseteq G_{i'}$, for these facilities respectively. Therefore, $G^{(C_i,C_{i'})}_{X_j \mapsto a} = G^{C_i}_{X_j \mapsto a} \cup G^{C_{i'}}_{X_j \mapsto 1- a}$.

Since, we established that $|S \cap G_i|=|S \cap G_{i'}|=1$, it must be the case that  $|S \cap G^{C_i}_{X_j \mapsto a}| = |S \cap G^{C_{i'}}_{X_j \mapsto 1- a}| =1$, as $|S \cap G^{(C_i,C_{i'})}_{X_j \mapsto a}| = 2$. However, this means that $|S \cap G^{C_i}_{X_j \mapsto 1- a}| = |S \cap G^{C_{i'}}_{X_j \mapsto a}| =0$, again due to $|S \cap G_i|=|S \cap G_{i'}|=1$, since $G^{C_i}_{X_j \mapsto 1- a} \subseteq G_i$ and $G^{C_{i'}}_{X_j \mapsto a} \subseteq G_{i'}$. But since the assignment group $G^{(C_i,C_{i'})}_{X_j \mapsto 1-a} = G^{C_i}_{X_j \mapsto 1-a} \cup G^{C_{i'}}_{X_j \mapsto a}$, and therefore $|S \cap G^{(C_i,C_{i'})}_{X_j \mapsto 1-a}|=0$, while the lower bound requirement of $G^{(C_i,C_{i'})}_{X_j \mapsto 1-a}$ is set to $1$. This contradicts the fact that $S$ is a feasible solution to $\cI_L$. Therefore, $S$ is a feasible solution to $\cI$.

\newpage
\section{A Polynomial Time Approximation Algorithm}
\label{app:polytime}

In this section, we design a polynomial time factor $\bigO(\log k)$-approximation for \CFRkMed with constant number of groups (when $t$ is constant). We formally state this result in the following theorem.

\polyapx*

For the sake of clarity, we focus on presenting our results for \CFRkMed. Although the extension to \CFRkMeans is straightforward, we will indicate the parts of the analysis that differ for \CFRkMeans.

On a high level, our algorithm works in two phases. In phase $1$, which is described in Section~\ref{ss:treeemb}, we embed the given instance $\Ical$ of \CFRkMed in general metrics into another instance $\Ical'$ of \CFRkMed on a special metric, which we call, \emph{\met} metric, such that the optimal cost of $\cI'$ is at most a constant factor of the optimal cost of $\cI$.
In  phase $2$, described in Section~\ref{ss:apxcliquestar}, we design a polynomial time $\bigO(\log k)$-approximation for \CFRkMed on \met metrics, thus, yielding $\bigO(\log k)$-approximation for the \CFRkMed in general metrics. We note that the transformations of our algorithm only modifies the underlying metric of $\cI$, leaving everything else untouched, i.e., the instance $\cI'$ is exactly same as $\cI$, except that the underlying metric $d'$ in $\cI'$ is different than the metric $d$ of $\cI$. 

\xhdr{Preliminaries}
For an instance $\cI$ of \CFRkMed, we assume that the metric $d$ is specified by a weighted graph $H$ on a vertex set $V$ such that  $V \supseteq C \cup F$, and $d$ corresponds to the shortest path metric on $H$.
Furthermore, we say $d$ is \emph{\kmet}, if $H$ consist of a $k$-clique $K_k$ and the renaming $n-k$ vertices are connected to exactly one vertex in $K_k$, i.e., they are pendant vertices to the clique nodes.
We say that a given instance $\cI$ of fair-range clustering is defined over a tree (or \kmet) metric if $d$ is a tree (\kmet, resp.).

\subsection{Embedding general metric spaces into \kmet metric}\label{ss:treeemb}
In this section, we show the following result that allows us to embed the instace $\cI$ using  an $\eta$-approximation algorithm for $k$-median into \kmet metric such that 
the optimal cost of the embedded instance is at most $\bigO(\eta)$ times the optimal cost of $\cI$. 

\cliquestaremb*
\begin{proof}
Given an instance $\Ical$ of \CFRkMed in metric space $d$, we first obtain $S\subseteq F$ by running algorithm $\Acal$ on $\cI$. Note that, $S$ might be infeasible for $\cI$, as $\Acal$ is an $\eta$-approximation algorithm for \kMed. However, it holds that $\cost_{\cI}(S) \le \eta\cdot \opt(\cI)$, since $\opt(\cI)$ is at most the optimal cost of \kMed  on $\cI$.
Using $S$, we construct the new metric $d'$ as follows. For every $s\in S$ and for every client $c\in C$ that is assigned to $s$, add an edge $(c,s)$ with weight $d(c,S)$. 
Next, for every facility $f \in F$, add an edge $(f,s_f)$, where $s_f \in S$ is the closest facility to $f$ in $S$, with weight $d(f,s_f)$.
Finally, for every $s \neq s' \in S$, add edge $(s,s')$ with weight $d(s,s')$. Let $H'$ be the resultant graph. Then, $d'$ is defined as the shortest path metric on $H'$. Note that $d'$ is \kmet metric.

Consider an optimal solution $O$ to $\cI$.
For a client $c \in C$,  let $o_c \in O$ be the facility serving $c$ in  $O$. 
Similarly,  for $c\in C$,  let $s_c \in S$  be the facility serving $c$ in $S$, and let $s_{o_c} \in S$ be the closest facility to $o_c$ in $S$.
Then, we have the following claim.
\begin{claim}\label{cl:optind1}
For every client $c \in C$ and the corresponding facilities $o_c \in O$ and $s_c \in S$, we have that
\[ d'(c,o_c) \leq 4 \cdot d(c,s_c) + 3 \cdot d(c,o_c).\]
\end{claim}
\begin{proof}
The scenario of this claim is illustrated in Figure~\ref{fig:polytime-apx}. In the metric space $d'$,  we have that,
%
\begin{align*}
d'(c,o_c) & = d(c,s_c) + d(s_c,s_{o_c}) + d(s_{o_c},o_c) \\
              & \overset{(i)}{\le} d(c,s_c) + d(s_c,c) + d(c,o_c) + d(o_c,s_{o_c}) + d(s_{o_c},o_c) \\
			  &= 2d(c,s_c) + d(c,o_c) + 2d(s_{o_c},o_c)\\
              & \overset{(ii)}{\leq} 2d(c,s_c) + d(c,o_c) + 2d(o_c,s_c)\\
              & \overset{(iii)}{\leq} 2d(c,s_c) + d(c,o_c) + 2(d(o_c,c) + d(c,s_c))\\
              &= 4 \cdot d(c,s_c) + 3 \cdot d(c,o_c),
\end{align*}
where, inequality ($i$) holds because of the triangle inequality $d(s_c,s_{o_c}) \le  d(s_c,c) + d(c,o_c) + d(o_c,s_{o_c})$, inequality ($ii$) holds because $s_{o_c}$ was the closest facility in $S$ to $o_c$, and hence $d(o_c,s_{o_c}) \le d(o_c,s_c)$, and finally inequality ($iii$) holds because of the triangle inequality $d(o_c,s_c) \le d(o_c,c) + d(c,s_c)$. Concluding this claim.
\end{proof}

\begin{figure}[t]
  \begin{minipage}[c]{0.45\textwidth}
    \includegraphics[width=\linewidth]{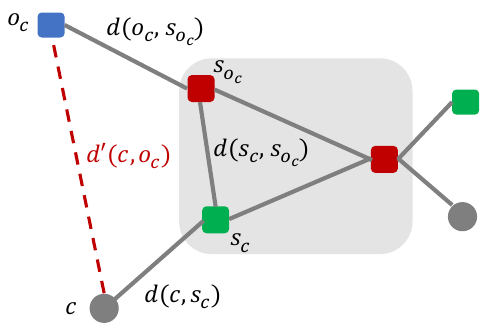}
  \end{minipage}
  \hfill
  \begin{minipage}[c]{0.5\textwidth}
    {\caption{\label{fig:polytime-apx}
    An illustration of the clique-star metric $d’$ for a client $c$.
    Facilities are represented as squares and clients as circles. The shaded area highlights the cluster centers $S$ selected by the $\eta$-approximation algorithm for vanilla \kMed. In this example, client $c$ is assigned to facility $s_c \in S$, while in the optimal solution, it is served by facility $o_c$, which is connected to center $s_{o_c} \in S$ in $d'$. Our goal (Claim~\ref{cl:optind1}) is to bound the rerouting cost $d’(c, o_c)$ in terms of distances in the original metric space $d$.
    }}
  \end{minipage}
\end{figure}

By summing the guarantee of Claim~\ref{cl:optind1} over all $c \in C$, we obtain,
\begin{align*}
 \cost_{\Ical'}(O) =
\sum_{c \in C} d'(c,o_c) & \leq  4 \cdot \sum_{c \in C} d(c,s_c) + 3 \cdot \sum_{c \in C} d(c,o_c) \\
 & \overset{(iv)}{\leq} 4 \eta \cdot \opt(\cI) + 3 \cdot \opt(\cI) \\
 & \leq (4 \eta + 3)\cdot \cost_{\cI}(O),
\end{align*}
where, inequality~($iv$) holds because $\sum_{c \in C} d(c,s_c) = \cost_{\cI}(S)  \le \eta\cdot \opt(\cI)$, due to algorithm $\Acal$, and $\sum_{c \in C} d(c,o_c) = \opt(\cI)$. Finally, note that $O'$ is an optimal solution to $\cI'$, and hence, we have that $\cost_{\cI'}(O') \le \cost_{\cI'}(O)$, as desired.
Since $\Acal$ runs in polynomial time, and construction of \kmet $d'$ can be done in polynomial time, the overall running time is $n^{\bigO(1)}$, which concludes the proof of lemma.
\end{proof}

\subsection{$O(\log k)$-approximation algorithm for \kmet metrics}\label{ss:apxcliquestar}
In this section, we show a $O(\log k)$-approximation algorithm for \CFRkMed on \kmet metrics. Towards this, in section~\ref{ss:kmettotree}, we show how to embed \kmet metric instance of \CFRkMed into tree, at a cost of a multiplicative factor $O(\log k)$ in the cost. Finally, in section~\ref{ss:treedp}, we show how to solve \CFRkMed exactly on trees using a polynomial time dynamic program, resulting in a $\bigO(\log k)$-approximation for \CFRkMed on \kmet metrics.

\subsubsection{Embedding \kmet metric into trees}\label{ss:kmettotree}

\treemetric*
\begin{proof}
	Recall that the \kmet metric $d'$ has a complete graph $Q$ on $k$ nodes, while remaining $n-k$ vertices are pendant to the nodes of $Q$. According to the seminal result of \citet{bartal1998onapproximating}, any graph with $k$ nodes can be embedded into a tree metric  on  $\bigO(k)$ nodes with distortion at most a $\bigO(\log k)$  in the distances. \footnote{The original construction of Bartal is randomized, however, as mentioned in\cite{adamczyk2019constant}, this construction can be derandomized.} 
	We use this tree embedding result to obtain $\cI''$ on a tree metric $d''$,  by simply replacing the complete graph $Q$ of \kmet metric $d'$ by the tree $T$ obtained from~\cite{bartal1998onapproximating}.
   Without loss of generality,  we assume that all the nodes of $Q$ of $d'$ are mapped to leaves of the $T$. By relabling the vertices of $T$, we assume that a vertex $u$ in $Q$ of $d'$ is mapped to $u$ in $T$ of $d''$. 
    Note that, for any pair $u,v \in Q$, we have $d'(u,v) \le d''(u,v) \le \bigO(\log k) \cdot d'(u,v)$. 

	Now, conside any subset $S'$ of facilities. 	We claim that $ \cost_{\Ical'} (S') \leq  \cost_{\Ical''}(S') \leq  \bigO(\log k) \cdot \cost_{\Ical'}(S')$. Since $d''$ dominates $d$, we have that  $ \cost_{\Ical'} (S') \leq  \cost_{\Ical''}(S')$. Furthermore, consider a client $c \in C$ and let $f \in S'$ be the facility serving $c$. Let $q_c , q_{f_c} \in Q$ be the nodes of $Q$ that are closest to $c$ and $f_c$ in $Q$, respectively. Then, since $d''(c,f_c) = d'(c,q_c) + d''(q_c,q_{f_c}) + d'(q_{f_c},f)$, we have,
	\[
	d''(c,{f_c}) \leq d'(c,q_c) + \bigO(\log k) \cdot d'(q_c,q_{f_c}) + d'(q_f,f) \leq \bigO(\log k) \cdot d'(c,{f_c}),
	\]
	since $d'(c,{f_c})=d'(c,q_c) + d'(q_c,q_{f_c}) + d(q_{f_c},f_c)$.

	By summing the above equation over all $c \in C$, we obtain:
	\begin{align*}
		\cost_{\cI''} (S') = \sum_{c \in C} d''(c,f_c) \le  \bigO(\log k) \sum_{c \in C}  d'(c,f_c)) = \bigO( \log k) \cdot \cost_{\Ical'}(S').
	\end{align*}
	The running time of the algorithm is polynomial, as the algorithm for embedding to tree metrics by \citet{fakcharoenphol2004atight} runs in polynomial time and other operations also  take polynomial time. This concludes the proof of lemma.
\end{proof}

\subsubsection{An exact dynamic programming algorithm for tree metrics}\label{ss:treedp}
The final ingredient for our approach is an exact algorithm for solving \CFRkMed on tree metrics. Our approach is inspired by the classical dynamic-programming algorithm for the vanilla $k$-median problem on tree metrics due to \citet{tamir1996pmedian} and its extension to the capacitated $k$-median problem by \citet{adamczyk2019constant}.

Let $T(e, \kappavec, b)$ denote the minimum clustering cost for the subtree rooted below edge $e$, where $\kappavec = (\kappa_1,\dots,\kappa_t)$ specifies that exactly $\kappa_i \in [t]$ facilities from group $G_i$ are opened, and $b \in \{-n, \dots, n\}$ represents the number of clients routed (in or out) through edge $e$. Here, $b < 0$ (or $b > 0$ resp.) indicates a capacity deficit (or surplus) of $|b|$ clients---meaning that $-b$ clients are routed upward (or $+b$ clients downward) through $e$. We now prove the following claim.

\exacttree*
\begin{proof}
Given an instance of \CFRkMed on the tree metric $d''$, we first preprocess the tree into a rooted binary tree using standard techniques, placing all clients and facilities are placed at the leaves. The root has one child, and all internal vertices have left and right children.
This can be done adding dummy nodes and zero-distance edges ensuring that the total number of nodes and edges remains in $\bigO(n)$, and the distances (for brevity, the metric $d''$) remain unchanged. For each non-root node $v$, let $e$ denote the edge connecting $v$ to its parent, and $e^\ell, e^r$ the edges connecting $v$ to its left and right children, respectively. The root node is connected by a special edge $e^\partial$. See Figure~\ref{fig:binary-tree} for an illustration.

On this transformed tree, we proceed bottom-up, starting from edges connected to leaf nodes and moving towards the root. In the base case, if an edge $e$ connects to a client, we set $T(e, \kappavec, b) = \infty$ for all $\kappavec \in [k]^t$ and $q \in \{-n, \dots, n\}$, as no facility is available to serve the client directly at this point. 
If an edge $e$ connects to a facility $f$ belonging to group $G_i$ with capacity $\zeta(f)$, then for $\kappavec=(0,\dots,1,\dots,0)$ with $\kappa_i=1$ and all other other entries zero, and for $b \in \{0,\dots,\zeta(f)\}$ we set $T(e,\kappavec,b) = d''(e) \cdot |b|$; otherwise, we set $T(e,\kappavec,b) = \infty$.

For each edge $e$, computing $T(e,\kappavec,q)$ would requires minimizing over all valid decompositions such that $\kappavec^\ell + \kappavec^r = \kappavec$ (element-wise) and $b^\ell + b^r = b$, where $\kappavec^\ell, \kappavec^r \in [k]^t$. After identifying the tuples $\kappavec^\ell, \kappavec^r, b^\ell, b^r$, we set:
\begin{equation} \label{eq:opt-assign}
T(e,\kappavec,b) = \min_{\substack{\kappavec^\ell+\kappavec^r=\kappavec \\ b^\ell + b^r = b}} \{T(e^\ell,\kappavec^\ell,b^\ell) + T(e^r, \kappavec^r,b^r)\} + |b| \cdot d''(e),
\end{equation}
where $d''(e)$ is the length of edge $e$ in the tree. The final solution is obtained by minimizing over the pseudo-root edge $e^\partial$ over all feasible $\kappavec \in [k]^t$ satisfying the fair-range constraints, \ie, $\alphavec \leq \kappavec \leq \beta$, as follows:
\begin{equation}
\min_{\substack{\kappavec \in [k]^t \\ \text{s.t.~} \alphavec \leq \kappavec \leq \betavec}} 
T(e^\partial, \kappavec, 0)
\end{equation}
The optimal subset of facilities can be recovered by storing the corresponding facility subset, rather than just the clustering cost, at each entry $T(e,\kappavec,b)$.

\xhdr{Optimality of the solution} We prove optimality by induction. In the base case, for edges connecting to leaf nodes, each entry $T(e, \kappavec, b)$ stores the optimal cost for that corresponding configuration.
For the induction step, assume that $T(e^\ell,\kappavec^\ell,b^\ell)$ and $T(e^r, \kappavec^r,b^r)$ store the optimal costs for all configurations $\kappavec^\ell, \kappavec^r \in [k]^t$ and $b^\ell, b^r \in \{-n,\dots,+n\}$. In Eq.\ref{eq:opt-assign}, we iterate over all valid decompositions satisfying $\kappavec^\ell + \kappavec^r = \kappavec$ and $b^\ell + b^r = b$, selecting the tuple $(\kappavec^\ell,\kappavec^r, b^\ell, b^r)$ that minimizes the total cost. 
Note that $b$ clients are routed through edge $e$ (either inwards or outwards) via edge $e$, incurring an additional cost of $|b| \cdot d''(e)$. 
Thus, the value stored in $T(e,\kappavec,q)$ is guaranteed to be optimal among all such combinations.

\xhdr{Running time} For each $e$, computing all configuration $(e,\kappavec,q)$ and their corresponding tuples $(\kappavec^\ell, \kappavec^r,b^\ell,b^r)$ minimizing Eq~\ref{eq:opt-assign} requires $\bigO(k^{2t} \cdot n)$ time. Since the binary tree has $\bigO(n)$ edges, the total running time of the algorithm is $k^{2t} \cdot n^{\bigO(1)}$, concluding the proof.
\end{proof}

\begin{figure}[t]
  \begin{minipage}[c]{0.5\textwidth}
    \includegraphics[width=\linewidth]{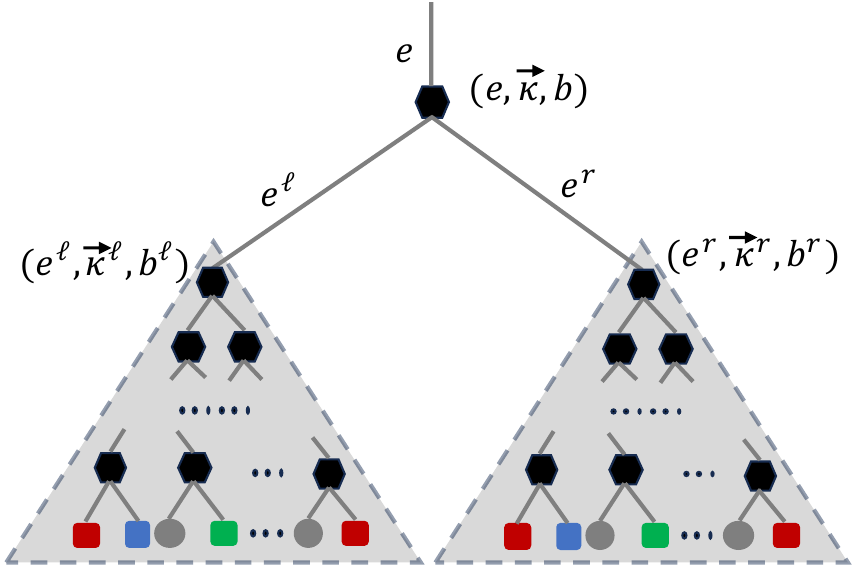}
  \end{minipage}
  \hfill
  \begin{minipage}[c]{0.48\textwidth}
    {\caption{\label{fig:binary-tree}
   An illustration of the dynamic programming algorithm. Clients are represented as circles, facilities as squares, and internal (dummy) nodes as hexagons. The color of each facility (red, blue, green) indicates its group membership. In the dynamic programming step to compute $T(e, \kappavec, b)$, we find the minimum cost over all decompositions of the subtree solutions connected via edges $e^\ell$ and $e^r$, such that $\kappavec^\ell + \kappavec^r = \kappavec$ for all $\kappavec^\ell, \kappavec^r \in [k]^t$ and $b^\ell + b^r = b$. There is an additional cost for re-routing $|b|$ clients, which is $d''(e) \cdot |b|$, where $d''(e)$ is the distance of edge $e$ in the tree metric $d''$.
    }}
  \end{minipage}
\end{figure}

\subsection{Proof of Theorem~\ref{thm:polyapx}}
Our algorithm for \CFRkMed (and \CFRkMeans) is as follows. Let $\cI$ be the input instance with metric $d$. First, we use Lemma~\ref{lemma:cliquestaremb} on $\cI$ to obtain an instance $\cI'$ on a \kmet metric $d'$, using $\bigO(1)$-approximation algorithm for $k$-median (or $k$-Means)~\cite{cohenaddad2025a2epsilon, aryalocal, ahmadian2019better}. Then, we obtain $\bigO(\log k)$-approximate ($\bigO(\log^2 k)$-approximate, resp.) solution $S$ to $\cI'$ using the algorithm of Section~\ref{ss:apxcliquestar}. Finally, we return $S$ as the solution. We claim that $S$ is an $\bigO(\log k)$-approximate ($\bigO(\log^2 k)$-approximate, resp.) solution to $\cI$. 

First, note that $S$ is a feasible solution to $\cI$ due to the guarantee of Section~\ref{ss:apxcliquestar}, and due to the fact that $\cI$ and $\cI'$ are same, except with different metrics. Then, note that,   $\opt(\cI) \leq \cost_{\cI}(S) \leq \cost_{\cI'}(S')$, because the metric $d'$ dominates $d$, i.e., $d'(u,v) \geq d(u,v)$, for all pairs $u,v \in C \cup F$. We finish the proof by observing that, for \CFRkMed, we have
\[
\cost_{\cI'}(S) \leq\bigO(\log k) \cdot \opt(\cI') \leq  \bigO(\log k) \cdot  \opt(\cI),
\]
due to Lemma~\ref{lemma:cliquestaremb}, since $\eta = O(1)$. Similarly, for \CFRkMeans, we have,
\[
\cost_{\cI'}(S) \leq\bigO(\log^2 k) \cdot \opt(\cI') \leq  \bigO(\log^2 k) \opt(\cI),
\]
Finally, all the parts of our algorithm, individually run in polynomial time, concluding the proof.

\newpage
\section{A Fixed Parameter Tractable Time Approximation Algorithm}
\label{app:fpt-apx}

In this section, we present $(3+\epsilon)$- and $(9+\epsilon)$-approximation algorithm for \CFRkMed and \CFRkMeans, respectively, running in $\fpt(k)$-time when the number of groups $t$ is constant. We formally state the result below.

\fptapx*

Our algorithm proceeds in three phases. Given an instance of \CFRkMed (or \CFRkMeans), Phase~1 (Section~\ref{app:fpt-apx:coresets}) reduces the number of clients to a small, weighted set called a coreset, such that the clustering cost on the coreset approximates that of the original instance. In Phase 2 (Section~\ref{app:fpt-apx:feasiblecp}), we transform the given instance with intersecting groups to multiple simplified instances with $k$-disjoint facility sets and replace clients with the coreset. In these simplified instances, selecting exactly one facility from each group satisfies the fair-range constraints, by construction. However, choosing arbitrary facilities may result in unbounded clustering cost relative to the optimal.
In Phase~3 (described in Section~\ref{app:fpt-apx:opg}), we design an approximation algorithm for each simplified instance using the leader-guessing framework. Applying this framework directly would not ensure capacity and fairness constraints are met, so we modify this framework to properly account for these requirements. By combining all three phases, we obtain approximation algorithms for \CFRkMed (or \CFRkMeans) with the claimed guarantees.
Throughout, we describe our algorithm in the context of \CFRkMed, noting only the necessary modifications for extending the results to \CFRkMeans.

\subsection{Constructing coresets for \CFRkMed and \CFRkMeans}
\label{app:fpt-apx:coresets}

Our first step is to reduce the number of clients using coresets---a small, weighted subset of clients that approximates the clustering cost of the original client set within a bounded factor. To accommodate this, we consider each client to be associated with a weight. We now formally define the weighted variants of the capacitated fair-range $k$-median and $k$-means problems. We remark that, the unweighted variant corresponds to each client having unit weight.
    
\begin{definition}[The weighted capacitated fair-range $k$-median (and $k$-means) problem]
An instance $\CFRkClustIns$ of the capacitated fair-range $k$-median (or $k$-means) problem can be extended to its weighted version, denoted as $\WCFRkClustIns$, by associating each client $c \in C$ with a weight $\omega(c) \geq 1$. The clustering cost of a solution $S \subseteq F$ is then defined as
$\sum_{c \in C} \omega(c) \cdot d(c,\rho(c))$
for the $k$-median objective, and 
$\sum_{c \in C} \omega(c) \cdot d(c,\rho(c))^2$
for the $k$-means objective.
\end{definition}

When clients are associated with positive real weights, it may be necessary to assign each client fractionally to multiple cluster centers, via an assignment function  $\mu: C \times S \rightarrow \mathbb{R}_{\geq 0}$. This variant of the problem is particularly relevant to us, as the coresets we construct will be weighted---where the weight, loosely speaking, represents the number of original clients that a coreset client stands in for. To properly handle these weighted clients, we require the flexibility to assign their weights fractionally across multiple centers, in a way redistributing portions of the original client mass. We formally define the fractional weighted variant of the capacitated fair-range clustering problem as follows.

\begin{definition}[The fractional weighted capacitated fair-range $k$-median (and $k$-means) problem]
An instance $\WCFRkClustIns$ of the weighted capacitated fair-range clustering problem can be extended to its factional variant, where the $c \in C$ with weight $\omega(c)$ can be assigned fractionally via a function $\mu: C \times S \rightarrow \RR_{\geq 0}$ such that
$\mu$ is a proper assignment for all clients, \ie, for all $c \in C, \sum_{f \in S} \mu(c,f) = \omega(c)$. The clustering cost of a solution $S \subseteq F$ is $\fraccost(C,S) = \sum_{c \in C, f \in S}\,\mu(c,f)\,d(c,f)$ for $k$-median and $\fraccost(C,S) = \sum_{c \in C, f \in S}\,\mu(c,f)\,d(c,f)^2$ for $k$-means.
\end{definition}

With necessary definitions in place, let us formally define coresets for \CFRkMed and \CFRkMeans.

\begin{definition}[Coresets]
For a given instance $\CFRkClustIns$ of \CFRkMed (or \CFRkMeans), a coreset is a pair $(W, \omega)$, where $W \subseteq C$ is a subset of clients and $\omega: W \rightarrow \RR_{\geq 0}$ is a weight function assigning non-negative weights to clients in $W$. The coreset approximates the clustering cost of the full client set $C$ such that, for every subset of centers $S \subseteq F$ of size $|S| = k$, and for some small $\epsilon_1 > 0$, the following holds:
\[
\fraccost(W, S) \in (1 \pm \epsilon_1)\, \cost(C, S).
\] 
\end{definition}

For any sufficiently small constant  $\epsilon_1 > 0$, \citet{cohenaddad2019on} constructed coresets of size $\bigO(k^2 \epsilon_1^{-3} \log^2 n)$ for capacitated $k$-median and $\bigO(k^5 \epsilon_1^{-3} \log^5 n)$ for capacitated $k$-means. The result is formally stated in the following theorem.

\begin{theorem}[\citet{cohenaddad2019on}, Theorem 11, 12]
For every $\epsilon_1 > 0$, there exists a randomized algorithm that, given an instance $(k,C,F,\zeta)$ of the capacitated $k$-median or $k$-means problem, computes a coreset $(W, \mu)$, where $W \subseteq C$ and $\mu: W \rightarrow \RR_{\geq 0}$ in $n^{\bigO(1)}$ time. The coreset has size $|W|= \bigO(k^2 \epsilon_1^{-3} \log^2 n)$ for the capacitated $k$-median problem and $|W|= \bigO(k^5 \epsilon_1^{-3} \log^5 n)$ for the capacitated $k$-means problem.
\end{theorem}

The above coreset applies to any subset of facilities $S \subseteq F$, regardless of whether it satisfies the fair-range constraints or not. Thus, a coreset constructed for the vanilla capacitated $k$-median (or $k$-means) problem can be directly used for our setting. We state this formally in the following corollary.

\begin{corollary}[Coresets for \CFRkMed and \CFRkMeans]
\label{cor:coreset}
For every $\epsilon_1 > 0$, there exists a randomized algorithm that, given an instance $\CFRkClustIns$ of \CFRkMed or \CFRkMeans, computes a coreset $(W, \mu)$, where $W \subseteq C$ and $\mu: W \rightarrow \RR_{\geq 0}$ in $\poly(n,k)$ time. The coreset has size $|W|= \bigO(k^2 \epsilon_1^{-3} \log^2 n)$ for \CFRkMed and $|W|= \bigO(k^5 \epsilon_1^{-3} \log^5 n)$ for \CFRkMeans.
\end{corollary}

\subsection{Enumerating feasible constraint patterns}
\label{app:fpt-apx:feasiblecp}
In this subsection, we reduce \CFRkMed (and \CFRkMeans) to multiple instances of a simpler problem variant, which has a restricted structure that allow us to design efficient algorithmic solutions. To do so, we adopt the framework of \citet{thejaswi2022clustering, thejaswi2024diversity} to our setting.
For simplicity, we describe our approach in the context of \CFRkMed, however, since finding a feasible solution is independent of the clustering objective, the approach naturally extends to \CFRkMeans.
We begin by introducing the notions of characteristic vector and \pattern.

\begin{definition}[Characteristic vector]
Given an instance $\CFRkClustIns$ of \CFRkMed with $\GG=\{G_i\}_{i \in [t]}$. A characteristic vector of a facility $f \in F$ is a binary vector $\charvec_f \in \{0,1\}^t$,  where the $i$-th index $\charvec_f[i]$ is set to $1$ if $f \in G_i$ and $0$ otherwise.
Consider the set $\{\charvec_f\}_{f\in F}$ of characteristic vectors of facilities in $F$. For each $\gammavec \in \{0,1\}^t$, let $\patternset{\gammavec} = \{f \in F: {\charvec}_f=\gammavec\}$ denote the set of all facilities with characteristic vector $\gammavec$. 
Finally, $\PP=\{E(\gammavec)\}_{\gammavec \in \{0,1\}^t}$ induces a partition on $F$.
\end{definition}

\begin{definition}[Constraint pattern]
Given a $k$-multiset $\EE=\{E(\gammavec_{i})\}_{i \in [k]}$, where each $E(\gammavec_{i}) \in \PP$, the \pattern associated with $\EE$ is the vector obtained by the element-wise sum of the characteristic vectors $\{\gammavec_{1},\dots,\gammavec_{k}\}$, \ie, $\sum_{i \in [k]} \gammavec_{i}$.
A \pattern is said to be feasible if $\alphavec  \leq \sum_{i \in [k]} \gammavec_{i} \leq \betavec$, where the inequalities are taken element-wise.
\end{definition}
In the following lemma, we establish that enumerating all feasible constraint patterns for \CFRkMed can be done in $\fpt(k,t)$-time.

\feasiblecp*
\begin{proof}
The partition $\PP$ can be constructed in $\bigO(|F|\,t)$ time, as there are at most $|F|$ facilities and computing each characteristic vector takes $\bigO(t)$ time. We then enumerate all possible $k$-multisets over $\PP$ and check if they form a feasible constraint pattern. Since there are \( \binom{|\PP| + k - 1}{k} \) such multisets, enumeration takes $\bigO(|\PP|^k t)$ time. Given $|\PP| \leq 2^t$ and $|F| \leq n$, the running time is $\bigO(2^{tk} \, t n)$.
\end{proof}

Given $(\gammavec_i)_{i \in [k]}$ be a feasible \pattern. 
By selecting exactly one facility from each $E(\gammavec_{i})$ arbitrarily for every $i \in [k]$ yields a feasible solution for \CFRkMed.
However, such an arbitrary selection may not guarantee a bounded approximation factor relative to the optimal solution. Therefore, we need to select facilities more carefully.
We remark that a feasible constraint pattern $(\gammavec_i)_{i \in [k]}$ may contain two (or more) identical vectors, \ie, $\gammavec_i = \gammavec_j$ for some $i \neq j$, when multiple facilities must be selected from the same group. In such cases, we duplicate facilities in $E(\gammavec_i)$, introducing infinitesimal small distortion $\epsilon_2 > 0$ to create a new group $E(\gammavec_j)$. As a result, we treat every collection $\EE$ corresponding to a feasible constraint pattern as consisting of disjoint groups.
This transformation reduces the problem to a simpler variant: the one-per-group weighted capacitated $k$-median (or $k$-means) problem with disjoint facility groups, where the goal is to choose exactly one facility from each group $E(\gamma_i)$, for every $i \in [k]$ and the clients are replaced by a coreset with associated weights. We now define this problem formally.

\begin{definition}[The one-Per-Group Weighted Capacitated $k$-Median (or $k$-means) problem]
\label{def:OPGWCkMed}
An instance $\OPGWCkMedIns$ of the One-Per-Group Weighted Capacitated $k$-Median (or $k$-means) problem is defined by a positive integer $k$, a coreset $(W, \omega)$, where $\omega: W \rightarrow \mathbb{R}_{\geq 0}$, a collection of disjoint facility groups $\EE = \{E_i\}_{i \in [k]}$, facility capacities $\zeta: \bigcup_{i \in [k]} E(\gamma_i) \rightarrow \mathbb{R}_{\geq 0}$. The task is to choose exactly one center from each group $E_i$, for all $i \in [k]$, and assign clients (fractionally) to the selected centers $S$ via an assignment function $\mu: W \times S \rightarrow \mathbb{R}_{\geq 0}$, such that: 
\squishlist
\item for all clients $c \in W$, $\sum_{f \in S} \mu(c, f) = \omega(c)$,
\item  for all selected facilities $f \in S$, $\sum_{c \in W} \mu(c, f) \leq \zeta(f)$.
\squishend
The objective is to minimize the clustering cost ${\sum_{c \in W, f \in S} \mu(c, f) \cdot d(c, f)}$ for $k$-median and ${\sum_{c \in W, f \in S} \mu(c, f) \cdot d(c, f)^2}$ for $k$-means. We succinctly denote these problems as \OPGWCkMedDis and \OPGWCkMeansDis, respectively.
\end{definition}

\subsection{An \fpt approximation algorithm for \OPGWCkMedDis and \OPGWCkMeansDis}
\label{app:fpt-apx:opg}

In this subsection, we give a $(3 + \epsilon_3)$-approximation algorithm for \OPGWCkMedDis, for any $\epsilon_3 > 0$. The proof for \OPGWCkMeansDis is similar and can be obtained by replacing the distances with squared distances, resulting in the claimed $(9 + \epsilon_3)$-approximation.

\begin{lemma} \label{lemma:opg-fpt-apx}
There exists a randomized $(3+\epsilon_3)$-approximation algorithm, for every $\epsilon_3 > 0$, for \OPGWCkMedDis in time $(\bigO(|W|\, k\, \epsilon_3^{-1} \log n))^{\bigO(k)} \cdot n^{\bigO(1)}$, where $|W|$ is the size of the coreset. With the same running time, there exists a randomized $(9+\epsilon_3)$-approximation algorithm for \OPGWCkMeansDis.
\end{lemma}
\begin{proof}
Given an instance $\OPGWCkMedIns$ of \OPGWCkMedDis with $\EE=\{E_i\}_{i \in [k]}$.
First we guess a set of $k$-leaders $L^*= (\ell_1^*, \dots, \ell_k^*) \subseteq W$ that will be closest to the facilities in the optimal solution $S^* = \{f_1^*, \dots, f_k^*\}$, with $\ell_i^*$ being closest to $f_i^*$. 
Additionally, we guess the radii $R^* = (r_1^*, \dots, r_k^*)$ corresponding to each of the $k$-clusters in the optimal solution. More precisely, the aspect ratio of the metric space is bounded in interval $[1,(n)^{\bigO(1)}]$ which we discretize to $[[(n)^{\bigO(1)}]_{\epsilon_3}]$, for some $\epsilon_3 >0$. So there are at most $[(n)^{\bigO(1)}]_{\epsilon_3} \le \lceil \log_{1+\epsilon_3} \Delta\rceil = \bigO(\epsilon_2^{-1} \log n)$ discrete possible radii.
Out of these $\bigO(\epsilon_3^{-1} \log n)$ possible radii, we guess the correct radius $r_i^*$ from the leader $\ell_i^*$ to the facility $f_i^*$ that will be serving $\ell_i^*$ in the optimal solution, \ie, $d(\ell_i^*,f_i^*) = r_i^*$.

Although we can guess the set of leaders $L^*$ and their corresponding optimal radii $R^*$ from the optimal solution $S^*$, multiple facilities may lie within each radius, and it remains unclear which one to choose. The selected facilities must satisfy the following conditions: 
($i$) exactly one facility must be chosen from each group $E_i$, 
($ii$) the number of assigned clients should not exceed the capacity of facilities, 
($iii$) the clustering cost must bounded with respect to the cost of the optimal solution,
($iv$) finding the facility set satisfying $i$, $ii$ and $iii$ must be in $\fpt(k,t)$ time.

\xhdr{One-per-group constraint}
As mentioned earlier, there can be multiple facilities belonging to different $E(\gamma)$ may be present within the radius $r^*_{i}$ from a leader $\ell^*_i$, so it is unclear facility belonging to which group we need to select for $\ell_i^*$. As there are $k$  groups, we can find this in $\bigO(k^k)$ time via brute-force enumeration. 
For the remainder of the proof, we assume that the correct group $E(\gamma_j)$ is known for each leader $\ell_i^*$, and without loss of generality, we assume that leader $\ell_i^*$ selects a facility from group $E(\gamma_i)$ for all $i \in [k]$. Specifically, we choose the facility $f_i \in E(\gamma_i)$ with the largest capacity within radius $r_i^*$ from $\ell_i^*$.
Note that $f_i$ and $f_i^*$ (the facility used in the optimal solution) may not belong to the same group. While we enforce $f_i \in E(\gamma_i)$, we make no assumptions about the group membership of $f_i^*$, which may be from any group.

\xhdr{Capacity requirements}
Since $f_i$ has the highest capacity among all candidates in $E_(\gamma_i)$ such that $d(\ell^*_i,f_i) \leq r^*_i$, it can serve at least as much capacity as $f_i^*$, the corresponding facility in the optimal solution, \ie, $\zeta(f_i^*) \leq \zeta(f_i)$. For the purpose of bounding the approximation ratio, we assume that for each $i \in [k]$, all clients fractionally assigned to $f_i^*$ in the optimal solution are reassigned to $f_i$, \ie, $\mu(c,f_i) = \mu^*(c, f_i^*)$. Such an assignment is valid because $\zeta(f_i^*) \leq \zeta(f_i)$.

\xhdr{Approximation factor}
Consider the illustration in Figure~\ref{fig:fpt-apx-factor}. The distance incurred when redirecting any client $c$, originally fractionally served by $f_i^*$, to the selected facility $f_i$ can be expressed as follows:
\begin{align*}
d(c,f_i) & \overset{(i)}{\leq} d(c,f_i^*) + d(f_i^*,\ell_i^*) + d(\ell_i^*,f_i)\\
&\overset{(ii)}{\leq}  d(c,f_i^*) + 2 \cdot (1 + \epsilon_3)\, r_i^*\\
&\overset{(iii)}{\leq} d(c,f_i^*) + 2 \cdot (1 + \epsilon_3)\, d(c,f_i^*)\\
&\overset{(iv)}{\leq} (3 + 2 \epsilon_3)\, d(c,f_i^*),
\end{align*}
where the inequality ($i$) follows from the triangle inequality by splitting the distance appropriately, 
the inequality ($ii$) holds because we have guessed the radius $r_i^*$ such that $d(\ell_i^*, f_i) \leq r_i^*$. The additive factor of $\epsilon_3$ is coming for the distortion in the distances due to discretization of radii.
The inequality ($iii$) follows from the choice of $\ell_i^*$ as the closest client to  $f_i^*$, implying  $d(c, f_i^*) \geq d(\ell_i^*, f_i^*) = r_i^*$, allowing us to replace $r_i^*$ with $d(c, f_i^*)$.
Since only $\mu(c,f_i^*)$ is the fraction of $c$ that is being served by $f_i^*$ in the optimal solution, which we will redirect to $f_i$ in our suboptimal solution, it holds that,
\[ \mu(c,f_i^*) \cdot d(c,f_i) \leq (3 + 2 \epsilon_3) \cdot \mu(c, f_i^*) \cdot d(c,f_i^*)\]
By summing over all clients $c \in W$, the total assignment cost can be bounded as follows.

\begin{align*}
\sum_{c \in W, f_i \in S}\, \mu(c,f_i^*) \cdot d(c,f_i) &\leq (3 + 2 \epsilon_3) \sum_{c \in C, f_i^* \in S^*}\, \mu(c,f_i^*)\cdot d(c,f_i^*) \\
\fraccost(C,S) & \leq (3 + 2 \epsilon_3)\cdot \opt 
\end{align*}

\xhdr{Running time analysis}
Brute-force enumeration over all $k$-multisets of clients in coreset $W$ can be done in $\bigO(|W|^k)$ time.
Since we assume that the aspect ratio of the metric space is bounded by a polynomial in  n , we can discretize the distances into at most $\bigO(\epsilon_3^{-1} \log n)$ distinct values using standard techniques. Enumerating all possible $k$-multisets of discretized radii to guess the optimal radii can thus be done in $\bigO((\epsilon_3^{-1 }\log n)^k)$ time.
Guessing the correct facility group $E(\gamma_i)^*$ for each leader $\ell_i^*$ can be done in $\bigO(k^k))$ time by computing the perfect matching.
Finally, for each guess, we need to compute the cost of the corresponding solution to select the one with the minimum cost. This computation can be performed in $n^{\bigO(1)}$ time.
Thus, the total running time of the algorithm is:
\[
\bigO\left(|W|^k \cdot (\bigO(\epsilon_3^{-1}\log n))^k \cdot \bigO(k^k) \cdot n^{\bigO(1)}\right) 
= \left(\bigO(|W|\, k\, \epsilon_3^{-1} \log n)\right)^{\bigO(k)} \cdot n^{\bigO(1)}.
\]
The proof of $9 + \epsilon_3)$-approximation algorithm for \OPGWCkMeansDis follows by replacing distances with squared distances. This concludes our proof.
\end{proof}

\begin{figure}[t]
  \begin{minipage}[c]{.3\textwidth}
    \includegraphics[width=\linewidth]{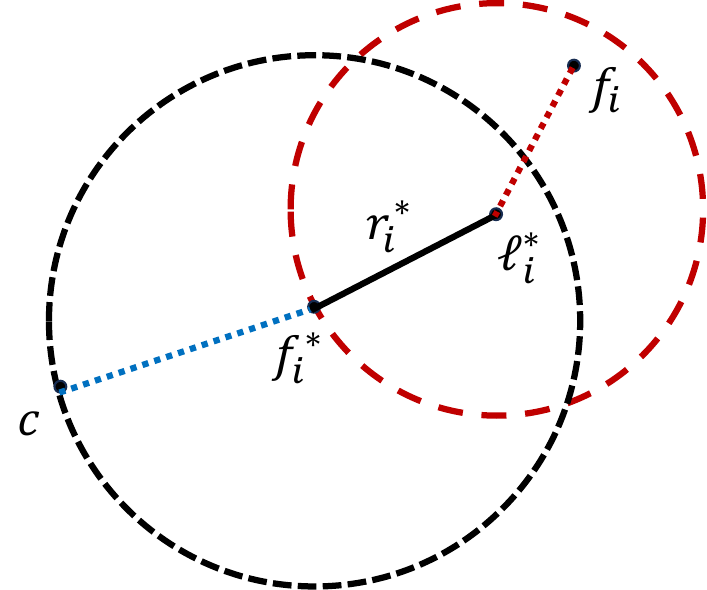}
  \end{minipage}
  \hfill
  \begin{minipage}[c]{0.65\textwidth}
    {\caption{ \label{fig:fpt-apx-factor}
    An illustration for bounding the approximation factor of the \fpt algorithm for \OPGWCkMedDis. Here, $\ell_i^*$ is the leader of cluster $i \in [k]$,$f_i$ is the facility closest to  $\ell_i^*$ within the guessed radius, and $f_i^*$ is the facility that (partially) serves both the client $c$ in the optimal solution $S^*=\{f_i^*\}_{i \in [k]}$. 
    If $\mu(c, f_i^*)$ denotes the fraction of $c$ served by $f_i^*$, we aim to bound the term $\mu(c, f_i^*) \cdot d(c, f_i)$ in terms of  $\mu(c, f_i^*) \cdot d(c, f_i^*)$, the former corresponds to the cost incurred in our approximate solution while the latter is corresponds to the cost incurred in the optimal solution if $c$ were to be (partially) served by $f_i^*$. }}
  \end{minipage}
\end{figure}

\subsection{Proof of Theorem~\ref{thm:fptapx}}

Our algorithm for solving \CFRkMed (or \CFRkMeans) proceeds as follows. For clarity, we describe the approach for \CFRkMed; the extension to \CFRkMeans is similar.
Given an instance \CFRkClustIns of \CFRkMed, we first construct a coreset $(W,\omega)$ for the client set $C$ using Corollary~\ref{cor:coreset}. Next, we partition $F$ into at most $\bigO(2^t)$ disjoint subsets $\PP = \{E(\gammavec)\}_{ \gammavec \in \{0,1\}^t }$, where each subset $E(\gammavec)$ contains facilities with a common characteristic vector $\gammavec \in \{0,1\}^t$. Using Lemma~\ref{lemma:feasiblecp}, we enumerate all feasible $k$-multisets of $\PP$, denoted as $\EE=\{E(\gammavec_{i})\}_{i \in [k]}$. To ensure the groups in \( \EE \) are disjoint, we duplicate facilities if needed and perturb distances slightly by $\epsilon_2 > 0$. This yields a set of at most $\bigO(2^{tk})$ instances of \OPGWCkMedDis, each defined as $J = (W,\EE,\zeta,\omega,k)$. 
For each instance $J$ of \OPGWCkMedDis, for any $\epsilon_3>0$, we compute a $(3 + \epsilon_3)$-approximate solution using in Lemma~\ref{lemma:opg-fpt-apx}, and choose the solution with minimum cost.

Recall that we may have duplicated facility groups in $\EE$ to ensure disjointness when constructing instances of \OPGWCkMedDis (or \OPGWCkMeansDis). If the same facility is selected more than once due to duplication, it would violate capacity constraints—since its capacity cannot be double-counted in the original instance. This issue is unique to the capacitated setting and does not arise in the uncapacitated case.
To handle this, we maintain a mapping $M$ that tracks all duplicates of each facility across groups. When a facility $f_i$ is selected, we set $M(f_i) = 1$. While selecting a facility for each leader $\ell^*_i$, we check if any duplicate of the candidate has already been chosen; if so, we skip it and select the next highest-capacity facility in the group.
This selection is valid and respects capacity constraints, as we always choose the highest-capacity facility available that has not been previously selected.

Finally, by choosing $\epsilon_1, \epsilon_2,\epsilon_3$ appropriately such that $\epsilon = \Theta(\epsilon_1  \epsilon_2  \epsilon_3)$, we have a $(3+\epsilon)$-approximation algorithm for \CFRkMed. 

\xhdr{Running time} Our algorithm for \CFRkMed invokes Lemma~\ref{lemma:opg-fpt-apx} at most $\bigO(2^{tk})$ times, with each invocation running in $(\bigO(|W|\, k\, \epsilon_3^{-1} \log n))^{\bigO(k)} \cdot n^{\bigO(1)}$ time. By Corollary~\ref{cor:coreset}, the coreset size is $|W| = \bigO(k^2 \epsilon_1^{-3} \log^2 n)$ for \CFRkMed and $|W|=\bigO(k^5 \epsilon_1^{-3} \log^5 n)$ for \CFRkMeans. Substituting the appropriate values for $|W|$ and choosing $\epsilon_1, \epsilon_2, \epsilon_3 > 0$ suitably yields the claimed overall running time of $(\bigO(2^t k \epsilon^{-1} \log n))^{\bigO(k)} \cdot n^{\bigO(1)}$. When $t$ is constant the running time is $(\bigO(k \epsilon^{-1} \log n))^{\bigO(k)} \cdot n^{\bigO(1)}$. This concludes our proof.
\hfill\qed



\end{document}